\def \VersionLong {}
\documentclass[a4paper,10pt]{llncs}
\usepackage[utf8]{inputenc}
\usepackage[english]{babel}

\usepackage[ruled,vlined,linesnumbered]{algorithm2e}
	\SetKwInOut{Input}{input}
	\SetKwInOut{Output}{output}
\SetAlFnt{\scriptsize}

\SetCommentSty{mycommfont}

\usepackage{subcaption}

\usepackage{paralist} % inline lists

\newenvironment{ienumeration}
	{\ifdefined\VersionLong\begin{enumerate}\else\begin{inparaenum}[\itshape i\upshape)]\fi}
	{\ifdefined\VersionLong\end{enumerate}\else\end{inparaenum}\fi}

\usepackage{amsmath} % for \text et al., for ``cases'' environment…
\usepackage{amssymb} % for \mathbb et al.

\usepackage{graphicx}
\graphicspath{{img/}{./}}

\ifdefined\VersionWithComments
	\usepackage{draftwatermark}
	\SetWatermarkText{draft}
	\SetWatermarkScale{15}
	\SetWatermarkColor[gray]{0.9}
\fi
\usepackage[svgnames,table]{xcolor}
\definecolor{darkblue}{rgb}{0.0,0.0,0.6}
\definecolor{darkgreen}{rgb}{0, 0.5, 0}
\definecolor{darkpurple}{rgb}{0.7, 0, 0.7}
\definecolor{darkblue}{rgb}{0, 0, 0.7}

\usepackage[
		colorlinks=true,
	\ifdefined \VersionWithComments
		pagebackref=true,
	\fi
		citecolor=darkgreen,
		linkcolor=darkblue,
		urlcolor=darkpurple,
	]{hyperref}

\usepackage[capitalise,english,nameinlink]{cleveref} % load after algorithm2e and hyperref
\crefname{line}{\text{line}}{\text{lines}} % to remove the capital

\usepackage{enumitem} % tighter itemize
\newcommand{\defProblem}[3]
{%	
\noindent\fcolorbox{black}{blue!15}{
	\begin{minipage}{.95\columnwidth}
		\textbf{#1 problem:}\\
		\textsc{Input}: #2\\
		\textsc{Problem}: #3
	\end{minipage}
}
	
	\smallskip
	
}

\ifdefined \VersionLong
	\newcommand{\LongVersion}[1]{\ifdefined\VersionWithComments{\color{red!60!black}#1}\else#1\fi}
	\newcommand{\ShortVersion}[1]{\ifdefined\VersionWithComments{\color{black!40}#1}\fi}
\else
	\newcommand{\LongVersion}[1]{\ifdefined\VersionWithComments{\color{black!40}#1}\fi}
	\newcommand{\ShortVersion}[1]{\ifdefined\VersionWithComments{\color{red!40!black}#1}\else#1\fi}
\fi

\usepackage{tikz}
\usetikzlibrary{arrows,automata,shapes} % shapes for cloud
\tikzstyle{every node}=[initial text=]
\tikzstyle{location}=[rectangle, rounded corners, minimum size=12pt, draw=black, fill=blue!10, inner sep=2pt]
\tikzstyle{splitlocation}=[location, rectangle split, rectangle split horizontal, rectangle split parts=2, thick]
\tikzstyle{locationa}=[location,thick]
\tikzstyle{locationb}=[location,draw=gray,thin]
\tikzstyle{invariant}=[draw=black, dotted, inner sep=1pt] % xshift=1em, 
\tikzstyle{final}=[double, fill=green!50!white]
\tikzstyle{urgent}=[fill=yellow, thick, dotted] % draw=red, very thick

\ifdefined \VersionWithComments
	\usepackage{marginnote}
	\newcommand{\marginX}{\marginnote{\huge{\quad\quad\textbf{!}\quad\quad}}}
	\newcommand{\ea}[1]{\mbox{}{\color{blue}\marginX{}\textbf{[\'Etienne}: #1]}}
	\newcommand{\jaco}[1]{\mbox{}{\color{orange}\marginX{}\textbf{[Jaco}: #1]}}
	\newcommand{\lp}[1]{\mbox{}{\color{purple}\marginX{}\textbf{[Laure}: #1]}}
	\newcommand{\vb}[1]{\mbox{}{\color{green!50!black}\marginX{}\textbf{[Vincent}: #1]}}
	\newcommand{\instructions}[1]{{\color{red}\marginX{}{[\textbf{Instructions:} ``#1'']}}}
	\newcommand{\reviewer}[2]{\mbox{}{\color{red}\marginX{}\textbf{[Reviewer #1}: ``#2'']}}
	\newcommand{\todo}[1]{\mbox{}{\color{red}{\marginX{}\textbf{TODO}\ifx#1\\\else:\ \fi #1}}} % here, ``\\'' stands for ``empty''
\else
	\newcommand{\instructions}[1]{}
	\newcommand{\ea}[1]{}
	\newcommand{\jaco}[1]{}
	\newcommand{\lp}[1]{}
	\newcommand{\vb}[1]{}
	\newcommand{\reviewer}[2]{}
	\newcommand{\todo}[1]{}
\fi

\ifdefined\VersionWithComments
	\usepackage[pagewise]{lineno} % switch, modulo
	\linenumbers
	
\fi

\newcommand{\init}{_0}
\newcommand{\styleSymbStatesSet}[1]{\ensuremath{\mathbf{#1}}}

\newcommand{\A}{\ensuremath{\mathcal{A}}}
\newcommand{\ALU}{\ensuremath{\A_\mathit{LU}}}
\newcommand{\Azeroinf}{\ensuremath{\A_{0,\infty}}}
\newcommand{\Actions}{\Sigma}
\newcommand{\action}{\ensuremath{a}}
\newcommand{\assign}{\leftarrow}

\newcommand{\C}{C}
\newcommand{\Cinit}{\C\init}
\newcommand{\Clock}{\mathbb{X}} % set of clocks
\newcommand{\ClockCard}{\ensuremath{{|\Clock|}}} % cardinality of clocks
\newcommand{\clock}{x} % clock
\newcommand{\clockval}{\nu_\Clock} % clock valuation
\newcommand{\ClocksZero}{\vec{0}}
\newcommand{\compOp}{\bowtie}
\newcommand{\compOpGeq}{\triangleright\!}
\newcommand{\compOpLeq}{\triangleleft}
\newcommand{\CTrue}{\mathbf{true}}
\newcommand{\CFalse}{\mathbf{false}}
\newcommand{\edge}{e}
\newcommand{\Edges}{E}
\newcommand{\longuefleche}[1]{\stackrel{#1}{\longrightarrow}}
\newcommand{\longueflecheRel}[1]{\stackrel{#1}{\mapsto}}

\newcommand{\flecheRel}{{\rightarrow}}
\newcommand{\Fleche}[1]{\stackrel{#1}{\Rightarrow}}
\newcommand{\grandn}{{\mathbb N}}
\newcommand{\grandq}{{\mathbb Q}}
\newcommand{\grandqplus}{\grandq_{+}} % \geq 0
\newcommand{\grandr}{\ensuremath{\mathbb R}}
\newcommand{\grandrplus}{\ensuremath{\grandr_{+}}} % \geq 0
\newcommand{\grandz}{{\mathbb Z}}
\newcommand{\guard}{g}

\newcommand{\invariant}{\mathcal{I}}
\newcommand{\K}{\ensuremath{K}}

\newcommand{\Kopt}{\ensuremath{\mathit{Opt}}}
\newcommand{\Topt}{\ensuremath{T_\mathit{opt}}}
\newcommand{\globaltime}{\ensuremath{\clock_\mathit{global}}}
\newcommand{\globaltimeparam}{\ensuremath{\param_\mathit{global}}}

\newcommand{\KFalse}{\bot}

\newcommand{\false}{\ensuremath{\mathit{False}}}

\newcommand{\loc}{\ensuremath{\ell}} % location % NOTE: nicer display than ``l'' (thanks Vincent)
\newcommand{\locinit}{\loc\init}
\newcommand{\Loc}{L} % set of locations
\newcommand{\locfinal}{\ensuremath{\loc_f}} % one particular final location
\newcommand{\lterm}{\mathit{lt}}
\newcommand{\Param}{\mathbb{P}} % set of parameters (P)
\newcommand{\param}{p} % parameter (p)
\newcommand{\ParamCard}{\ensuremath{{|\Param|}}} % number of parameters
\newcommand{\pval}{\nu_\Param} % parameter valuation
\newcommand{\pvalzeroinf}{\ensuremath{v_{0,\infty}}} % HACK à cause des notations de Vincent
\newcommand{\PZG}{\ensuremath{\mathcal{PZG}}} % state space, parametric zone graph
\newcommand{\R}{{\mathbb{R}}}
\newcommand{\Rgeqzero}{\R_{\geq 0}}

\newcommand{\sinit}{s\init} % initial set of states
\newcommand{\LocsTarget}{T} % subset of locations
\newcommand{\state}{\ensuremath{s}} % concrete state
\newcommand{\States}{S} % for LTS
\newcommand{\timelapse}[1]{#1^\nearrow}
\newcommand{\Runs}{\ensuremath{\mathit{Runs}}} % set of runs
\newcommand{\varrun}{\rho} % run

\newcommand{\duration}{\ensuremath{\mathit{duration}}} % duration of a run

\newcommand{\Passed}{\styleSymbStatesSet{P}}
\newcommand{\symbstate}{\ensuremath{\styleSymbStatesSet{s}}} % symbolic state
\newcommand{\SymbState}{\ensuremath{\styleSymbStatesSet{S}}} % set of symbolic states
\newcommand{\symbstateinit}{\symbstate\init} % initial symbolic state
\newcommand{\symbtrans}{{\Rightarrow}} % symbolic semantics transition relation
\newcommand{\Waiting}{\styleSymbStatesSet{W}}
\newcommand{\Queue}{\styleSymbStatesSet{Q}}
\newcommand{\sopt}{\ensuremath{\symbstate_\mathit{opt}}}
\newcommand{\timeval}{\ensuremath{t}} % time
\newcommand{\Tlow}{\ensuremath{\timeval}}

\newcommand{\resets}{R}
\newcommand{\project}[2]{\ensuremath{#1{\downarrow_{#2}}}}
\newcommand{\projectP}[1]{\ensuremath{#1{\downarrow_{\Param}}}}
\newcommand{\reset}[2]{\ensuremath{[#1]_{#2}}}
\newcommand{\valuate}[2]{\ensuremath{#2(#1)}}
\newcommand{\wv}[2]{#1|#2} % (w,v)

\newcommand{\MinTReach}{\ensuremath{\mathit{MinTimeReach}}}

\newcommand{\PMinTReach}{\ensuremath{\mathit{MinTimePTA}}}

\newcommand{\SynthPMinTReach}{\ensuremath{\mathit{MinTimeSynth}}}

\newcommand{\MinPReach}{\ensuremath{\mathit{MinParamReach}}}
\newcommand{\SynthMinPReach}{\ensuremath{\mathit{MinParamSynth}}}
\newcommand{\Reach}{\ensuremath{\mathit{Reach}}}
\newcommand{\stylealgo}[1]{\ensuremath{\textsf{#1}}}
\newcommand{\Synth}{\stylealgo{EFSynth}}
\newcommand{\MinParamSynth}{\stylealgo{MinParamSynth}}
\newcommand{\GetMin}{\stylealgo{GetMin}}

\newcommand{\MinTimeSynth}{\stylealgo{MinTimeSynth}}
\newcommand{\ExpMinTimeSynthNoIM}{\stylealgo{MTSynth-noRed}}
\newcommand{\ExpMinTimeSynth}{\stylealgo{MTSynth}}
\newcommand{\ExpMinTimeReach}{\stylealgo{MTReach}}

\newcommand{\ExpMinParamSynth}{\stylealgo{MPSynth}}
\newcommand{\ExpMinParamReach}{\stylealgo{MPReach}}
\newcommand{\ExpSynth}{\stylealgo{EFSynth}}
\newcommand{\Break}{\ensuremath{\mathbf{break}}}
\newcommand{\Continue}{\ensuremath{\mathbf{continue}}}
\newcommand{\Succ}{\mathsf{Succ}}
\newcommand{\Pop}{\ensuremath{\mathsf{Pop}}}
\newcommand{\Push}{\ensuremath{\mathsf{Push}}}
\DontPrintSemicolon

\newcommand{\imitator}{\textsf{IMITATOR}}

\ifdefined \VersionWithComments
 	\definecolor{colorok}{RGB}{80,80,150}
\else
	\definecolor{colorok}{RGB}{0,0,0}
\fi

\newcommand{\eg}{\textcolor{colorok}{e.\,g.}\xspace}
\newcommand{\ie}{\textcolor{colorok}{i.\,e.}\xspace}
\newcommand{\st}{\textcolor{colorok}{s.t.}\xspace}

\title{Minimal-Time Synthesis for\\Parametric Timed Automata%
\thanks{%
	\LongVersion{This is the author version of the manuscript of the same name published in the proceedings of the 25th International Conference on Tools and Algorithms for the Construction and Analysis of Systems (TACAS 2019).
	This version contains extended definitions, and all proofs.
	}
	This work is partially supported by the ANR national research program PACS
    (ANR-14-CE28-0002) and PHC Van Gogh project PAMPAS.
}}
\author{\'Etienne Andr\'e\inst{1,2,3}\orcidID{0000-0001-8473-9555}\thanks{Partially supported by ERATO HASUO
Metamathematics for Systems Design Project (No.\ JPMJER1603), JST.} \and
Vincent Bloemen\inst{4}\thanks{Supported by the 3TU.BSR project.}\and\\Laure
Petrucci\inst{1}\and Jaco van de Pol\inst{4,5}}
\institute{LIPN, CNRS UMR 7030, Université Paris 13, %\LongVersion{Sorbonne Paris Cité,\\}
Villetaneuse, France
\and JFLI, CNRS, Tokyo, Japan
\and National Institute of Informatics, Japan
 	\and
University of Twente, The Netherlands
	\and
University of Aarhus, Denmark}

\begin{document}

% For all page numbers, except p.1
\pagestyle{plain}

\maketitle

% HACK for Springer footnotes
\setcounter{footnote}{0}

% For page numbers on p.1
\thispagestyle{plain}

\ifdefined \VersionWithComments
	\textcolor{red}{\textbf{This is the version with comments. To disable comments, comment out line~3 in the \LaTeX{} source.}}
\fi

\begin{abstract}
Parametric timed automata (PTA) extend timed automata by allowing parameters
in clock constraints. Such a formalism is for instance useful when reasoning
about unknown delays in a timed system. Using existing techniques, a user can
synthesize the parameter constraints that allow the system to reach a specified
goal location, regardless of how much time has passed for the internal clocks.

We focus on synthesizing parameters such that not only the goal location is
reached, but we also address the following questions: \textit{what is the
minimal time to reach the goal location?} and \textit{for which parameter values can we achieve this?}
We analyse the problem and present
an algorithm that solves it. We also discuss and provide solutions for
minimizing a specific parameter value to still reach the goal.

We empirically study the performance of these algorithms on a benchmark set for
PTAs and show that \emph{minimal-time reachability synthesis} is more efficient to compute than
the standard synthesis algorithm for reachability.
\end{abstract}

% \LongVersion{\textbf{Note:
% to ease the reviewers' work, the text added to this long version (when compared to the 15-page version submitted at TACAS) is highlighted in color.
% }
% }
\ea{Dear all:
Since no appendix can be used, I suggest we use the macro \texttt{$\setminus$LongVersion\{\}} to \emph{add} sentences in the long version; we'll then put it on a Web page as suggested by TACAS.
According to my macros, if we use the \texttt{VersionLong} command line 6 of the \LaTeX{} code, all added text will be kept and highlighted.
If we comment it out, then all added text will simply disappear, resulting in a 15-page paper.
However, in comment mode (\texttt{VersionWithComments}, line~3 of the \LaTeX{} code) the added text still remains in gray.
}

\ea{If we need space, I suggest to drop all acknowledgments/supports (ANR, ERATO, PAMPAS, 3TU…), and to remove all my affiliations except the P13 one (without SPC of course).}

\instructions{TACAS 2019:
The length of research, case-study, and regular tool papers is limited to 15 pages plus 2 pages for references in the LNCS format. The length of tool-demonstration papers is limited to 6 pages in the LNCS format, including the bibliography.
\\
\emph{Appendices going beyond the above page limits are not allowed!} Additional (unlimited) appendices can be made available separately or as part of an extended version of the paper made available via arXiv, Zenodo, or a similar service, and cited in the paper. The reviewers are, however, not obliged to read such appendices.
\\
As in 2018, TACAS’19 will include artifact evaluation for all types of papers. For regular tool papers and tool demonstration papers, artifact evaluation is compulsory (see the TACAS’19 call for papers), for \emph{research and case-study papers, it is voluntary (papers with accepted artifacts will receive a badge).}
}

% \ea{hello}
% \jaco{hello}
% \lp{hello}
% \vb{hello}

% \ea{we can probably delete my multiple affiliations for the submitted version if we need space (same for the fundings)}

\todo{Check for naming consistency: minimal/minimum, optimal/optimum,
dot after paragraph/subsubsection title or
not. Is $\succ$ used correctly everywhere? Pairs with $(..)$ or
$\langle..\rangle$ (I prefer the latter one, but since that would change a lot
in the paper, I just changed it to $(..)$ ). Also, $\false$ or $\CFalse$ or
$\bot$?
minimum lower-bound time / minimum time lower-bound / minimum lower-time bound?
non-empty/non-empty}

%%%%%%%%%%%%%%%%%%%%%%%%%%%%%%%%%%%%%%%%%%%%%%%%%%%%%%%%%%%%
%%%%%%%%%%%%%%%%%%%%%%%%%%%%%%%%%%%%%%%%%%%%%%%%%%%%%%%%%%%%
\section{Introduction}\label{section:introduction}
%%%%%%%%%%%%%%%%%%%%%%%%%%%%%%%%%%%%%%%%%%%%%%%%%%%%%%%%%%%%
%%%%%%%%%%%%%%%%%%%%%%%%%%%%%%%%%%%%%%%%%%%%%%%%%%%%%%%%%%%%

 % NOTE: perhaps rename
\newcommand{\LA}{\ensuremath{\texttt{A}}}
\newcommand{\LB}{\ensuremath{\texttt{B}}}
\newcommand{\LC}{\ensuremath{\texttt{C}}}
\newcommand{\LD}{\ensuremath{\texttt{D}}}
\newcommand{\ALICE}{\ensuremath{\texttt{Alice}}}
\newcommand{\BOB}{\ensuremath{\texttt{Bob}}}
\newcommand{\DA}{\ensuremath{\mathit{D}_\mathit{1}}}
\newcommand{\DB}{\ensuremath{\mathit{D}_\mathit{2}}}

% TAs
\emph{Timed Automata (TA)}~\cite{AD94} extend finite automata with \emph{clocks}, for
instance to model real-time systems.
%\ea{To gain some space I'd remove the entire rest of this paragraph, useless to
%TACAS community, I guess}
%\vb{Fine by me}
\LongVersion{These clocks can be used to constrain
\emph{transitions} between two \emph{locations} with a \emph{guard}, \eg{} the transition can only
be taken if at least 5 time units have passed. Furthermore, aside from taking transitions,
it is possible to \emph{wait} some time at a location. This waiting time can also
be constrained by an \emph{invariant} associated with the location. Multiple clocks
can coexist and clocks may also be reset when taking a transition
(written as $\clock := 0$ for clock $\clock$).

% Optimal time reachability
}Timed automata allow for reasoning about temporal properties of the designed
system. In addition to reachability problems, it is possible to compute for TAs the
minimal or maximal time required to reach a specific goal location. Such a
result is valuable in practice, as it can describe the response time of a
system or it may indicate when a component failure occurs.

% PTA
It may not always be possible to describe a real-time system with a TA.
There are often uncertainties in the timing constraints, for instance how
long it takes between sending and receiving a message.
Optimising specific timing delays to improve the overall throughput of the system
may also be considered, as shown in \cref{ex:intro}. Such uncertainties can however be modelled using a \emph{parametric
timed automaton (PTA)}~\cite{AHV93}. A PTA adds parameters, or unknown
constants, to the TA formalism. By examining the reachability of a goal
location, the parameters get constrained and we can observe which parameter
valuations preserve the reachability of the goal location.

% Optimal time reachability for PTAs
This process, also called \emph{parameter synthesis}, is definitely useful for
analysing reachability properties of a system. However, this technique does
disregard timing aspects to some extent. Given the parameter constraints, it is
no longer possible to give clear boundaries on the time to reach the goal, as
this may depend on the parameter valuations. We focus on the 
parameter synthesis problem while reaching the goal location in minimal time,
as demonstrated in \cref{ex:intro}.

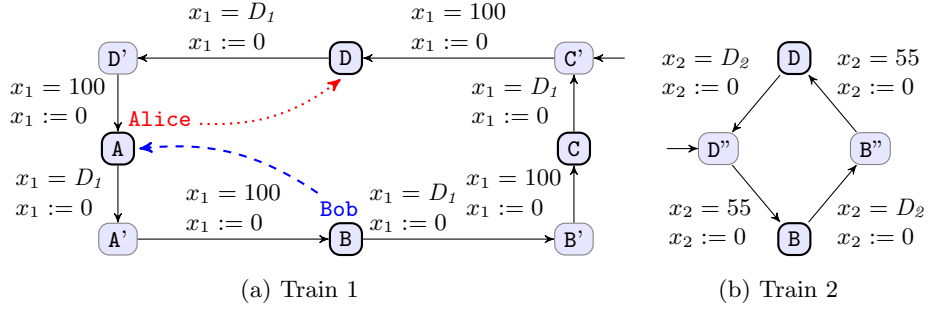
\begin{figure}[tb]
	{\centering
	%------------------------------------------------------------
	\begin{subfigure}[b]{0.65\textwidth}
		\begin{tikzpicture}[->, >=stealth']
			\newcommand{\yfactor}{1.2}
			\newcommand{\xfactor}{3}
			% locations
			\node[locationa] (la) at (0, 0) {\,$\LA$\,\nodepart{two}$\clock_1 \leq \DA$};
			\node[locationb] (laa) at (0, -1*\yfactor) {\,$\LA$'\,\nodepart{two}$\clock_1 \leq 100$};
			\node[locationa] (lb) at (\xfactor, -1*\yfactor) {\,$\LB$\,\nodepart{two}$\clock_1 \leq \DA$};
			\node[locationb] (lbb) at (2*\xfactor, -1*\yfactor) {\,$\LB$'\,\nodepart{two}$\clock_1 \leq 100$};
			\node[locationa] (lc) at (2*\xfactor, 0) {\,$\LC$\,\nodepart{two}$\clock_1 \leq \DA$};
			\node[locationb,initial right] (lcc) at (2*\xfactor, \yfactor) {\,$\LC$'\,\nodepart{two}$\clock_1 \leq 100$};
			\node[locationa] (ld) at (\xfactor, \yfactor) {\,$\LD$\,\nodepart{two}$\clock_1 \leq \DA$};
			\node[locationb] (ldd) at (0, \yfactor) {\,$\LD$'\,\nodepart{two}$\clock_1 \leq 100$};
			\path
			(la) edge[left] node[]{\begin{tabular}{l}$\clock_1=\DA$\\$\clock_1:=0$\end{tabular}} (laa)
			(laa) edge[below] node[above=-3pt]{\begin{tabular}{l}$\clock_1=100$\\$\clock_1:=0$\end{tabular}} (lb)
			(lb) edge[above] node[xshift=2pt,above left=-3pt]{\begin{tabular}{l}$\clock_1=\DA$\\$\clock_1:=0$\end{tabular}} (lbb)
			(lbb) edge[left] node[]{\begin{tabular}{l}$\clock_1=100$\\$\clock_1:=0$\end{tabular}} (lc)
			(lc) edge[left] node[]{\begin{tabular}{l}$\clock_1=\DA$\\$\clock_1:=0$\end{tabular}} (lcc)
			(lcc) edge[below] node[above=-3pt]{\begin{tabular}{l}$\clock_1=100$\\$\clock_1:=0$\end{tabular}} (ld)
			(ld) edge[below] node[above=-3pt]{\begin{tabular}{l}$\clock_1=\DA$\\$\clock_1:=0$\end{tabular}} (ldd)
			(ldd) edge[left] node[]{\begin{tabular}{l}$\clock_1=100$\\$\clock_1:=0$\end{tabular}} (la)
				;
			% Alice and Bob
			\node[] (bob) at (0.97*\xfactor,-1*\yfactor+0.4) {\textcolor{blue}{\BOB}};
			\node[] (alice) at (0.55, 0.4) {\textcolor{red}{\ALICE}};
            \path [->,draw,thick] (alice) edge[bend right=20,dotted,color=red] ([xshift=-2, yshift=-8]ld);
			\path [->,draw,thick] (bob) edge[bend right=20,dashed,color=blue] ([xshift=8, yshift=0]la);

			\renewcommand{\xfactor}{1}
			\newcommand{\xshift}{7.9}
			\node[locationa] (lld) at (\xshift+\xfactor, \yfactor) {\,$\LD$\,\nodepart{two}$\clock_2 \leq \DB$};
			\node[locationb, initial] (lldd) at (\xshift, 0) {\,$\LD$''\nodepart{two}$\clock_2 \leq 55$};
			\node[locationa] (llb) at (\xshift+\xfactor, -1*\yfactor) {\,$\LB$\,\nodepart{two}$\clock_2 \leq \DB$};
			\node[locationb] (llbb) at (\xshift+2*\xfactor, 0) {\,$\LB$''\nodepart{two}$\clock_2 \leq 55$};
			\path
				(lld) edge [above left] node[above left=-3pt]{\begin{tabular}{l}$\clock_2=\DB$\\$\clock_2:=0$\end{tabular}} (lldd)
				(lldd) edge [below left] node[below left=-3pt]{\begin{tabular}{l}$\clock_2=55$\\$\clock_2:=0$\end{tabular}} (llb)
				(llb) edge [below right] node[below right=-3pt]{\begin{tabular}{l}$\clock_2=\DB$\\$\clock_2:=0$\end{tabular}} (llbb)
				(llbb) edge [above right] node[above right=-3pt]{\begin{tabular}{l}$\clock_2=55$\\$\clock_2:=0$\end{tabular}} (lld)
				;
		\end{tikzpicture}
		\caption{Train 1}
		\label{fig:example:train1}
	\end{subfigure}
	\hfill
	\begin{subfigure}[b]{0.3\textwidth}
		\begin{tikzpicture}
            % empty picture, to keep caption position correct
		\end{tikzpicture}
		\caption{Train 2}
		\label{fig:example:train2}
	\end{subfigure}
	%------------------------------------------------------------
	}
	\caption{Train delay scheduling problem:
        {\ALICE} (depicted in \textcolor{red}{dotted red}), located at $\LA$, wants to
        go to station $\LD$.
        {\BOB} (depicted in \textcolor{blue}{dashed blue}), located at $\LB$, wants to
        go to\LongVersion{{} station}~$\LA$.
        By setting the train delays $\DA$ and $\DB$ for train 1
        and 2, make sure that both {\ALICE} and {\BOB} reach their target
        station in minimum total time.
		}
	\label{figure:example}
    \vspace{-1em}
\end{figure}

\begin{example}
\label{ex:intro}
Consider the example in \cref{figure:example}, which depicts a train network consisting
of two trains. Both trains share locations {\LB} and {\LD} (the stations platforms)
while locations $\LA',\LB',\LC',\LD',\LB'',$ and $\LD''$ represent a train travelling (tracks). The travel time for train 1 between any
two stations is 100, and~55 for train 2.
Train 1 stops at stations $\LA$, $\LB$, $\LC,$ and $\LD$, for time
$\DA$ (and train 2 stops for $\DB$ time units at $\LB$ and $\LD$).
Here, the train delays $\DA$ and $\DB$ are parameters and $\clock_1$ and $\clock_2$ are clocks.
Both clocks start at~0 and reset after every transition.
We assume that the trains use different tracks and changing trains 
at the platform of a station can be done in negligible time.

{\ALICE} is starting her journey from {\LA} and would like to go to {\LD}.
{\BOB} is
located at {\LB} and wants to go to {\LA}. Train 1 and/or 2 can be used to
travel, if both the train and the person are at the same location. Initially,
both {\ALICE} and {\BOB} wait for a train, since the initial positions of
train 1 and 2 are respectively {\LC'} and {\LD''}.
%\ea{Alice and Bob also have a dedicated automaton, isn't it? Did you build it? If so, can you add it in the appendix?}
%\vb{I did not (yet) build
%automata for Alice and Bob. An IMITATOR model would also be nice, probably.}
%\lp{In fact, with imitator I would have modelled \ALICE and \BOB with a variable indicating
%where they are : they don't do any actual action.}

We would like to set the train delays $\DA$ and $\DB$ in such
a way that the total time for {\ALICE} and {\BOB} to reach their target
location, \ie the PTA location for which {\ALICE} is at station $\LD$ and {\BOB} is at station
{\LA}, is minimal. The optimal solution is $\DA=25 \land \DB=15$, which leads
to a total time of 405 units\footnote{
    {\ALICE} waits for train 1 to reach {\LA} at time 225, then she hops on and
    exits the train on time 350 at {\LB}. There she can immediately take train
    2 and reach {\LD} at time 405.
    {\BOB} waits for train 2 to reach {\LB} at time 55 and takes this train.
    At time 125 he reaches {\LD} and can immediately hop on train 1. Bob
    reaches {\LA} at time 225.
}. Note that this is neither optimal for {\ALICE} (the fastest would be
$\DA=0 \land \DB=5$), nor optimal for {\BOB} ($\DA=10 \land \DB=0$).
\end{example}

% What is this paper about?
Note that in other instances, the time to reach a goal location may be an
interval, describing the lower- and upper-bound on the time.
This can be achieved in the example by changing the travel time from train 1 to be
between $95$ and $105$, by guarding the outgoing transitions from locations
$\LA'$, $\LB'$, $\LC'$ and $\LD'$ with $95 \leq \clock_1 \leq 105$
(instead of $\clock_1 = 100$).
%\vb{Is the invariant required?}
%\lp{the invariant on the locations is required. Otherwise we could stay forever in
%these locations. This is actually the case in figure 2, where we should put invariants
%on the locations corresponding to trains travelling on tracks, so that they eventually
%reach the next station. Maybe it's polish trains stuck for hours on the track in the
%middle of nowhere ;-( }
%\vb{Ok, but if we're only interested in runs that reach a target location, such
%cases should not make a difference. Anyway, in Figure 1 there is no mention of
%invariants, so either we probably want to have invariants at each location in
%there as well (which would make the figure difficult to read), or we should
%avoid mentioning them here and only talk about the outgoing transitions.. since
%a reader (such as me) may start to wonder why there are no invariants mentioned
%before}
%\lp{Changed the text a little bit to make things more clear.}
We focus on the lower-bound \emph{global time}, meaning that we look at the
minimal \emph{total} time passed in the system, which may differ from the clock
values as the clocks can be reset.
%\lp{This might be a bit confusing since global time is not described
%by a clock so cannot formally be reset.}

\medskip

In this paper we address the following problems:
\begin{itemize}[topsep=0pt]
    \item \emph{minimal-time reachability}: synthesizing \emph{a single}
        parameter valuation for which the goal location can be reached in
        minimal
        (lower-bound) time,
    \item \emph{minimal-time reachability synthesis}: synthesizing all parameter valuations such that
        the time to reach the goal location is minimized, and
%    \item \emph{minimal upper-bound-time synthesis}: synthesizing all
%        parameter valuations such that the upper-bound on the time to reach the goal
%        location is minimized (note that this is different from a ``maximum-time
%        reachability''), and
    \item \emph{parameter minimization synthesis}: synthesizing all parameter valuations
        such that a particular parameter is minimized and the goal location can still be
        reached (this problem can also address the
        \emph{minimal-time reachability synthesis problem} by adding a parameter to equal
        with the final clock value).
\end{itemize}
For all stated problems we provide algorithms to solve them and empirically
compare them with a set of benchmark experiments for PTAs, obtained from~\cite{Andre18FTSCS}.
Interestingly, compared to standard reachability and synthesis, minimal-time
reachability and synthesis is in general computed faster as fewer states
have to be considered in the exploration.
We also look at the computability and intractability of the problems for PTAs
and L/U-PTAs (PTAs for which each parameter only appears as a lower- or
upper-bound).

%\vb{Mention theoretical studies here as well?}
%\vb{What to do with our study on L/U-PTAs?}
%\todo{Possibly rewrite a bit and mention theoretical results, also mention
%parameter minimization reachability?}

\paragraph{Related work.}
\todo{Check this section and possibly improve/extend, possibly cite AM99}
%\todo{something about reachability for TAs have been well-studied, with refs}
%\todo{some related works to cite: \cite{BF01,ATP04,BLR05,ZNL16,ZNL16SAC}}

% ALTP04: optimal timed paths in weighted timed automata, parametric solution,
%        weights on locations and transitions, time spent in location * weight
% CY92: First consideration of timed optimal reachability
% BLR05: UPPAAL with linear parameter constraints, identification of L/U
%        automata
% BF01: minimal cost in uniformly priced timed automata (UPTA), Dijkstra,
%       optimal search order, branch-and-bound solutions for near-optimal
%       results 
% ZNL16SAC: Time optimal reachability with swarm verification, branch and
%           bound, randomized DFS/BFS, distributed
% ZNL16: Extension of above, partitioned state space, to divide work better and
%        handle larger state spaces
% NTY00: minimal-Time Reachability for Timed Automata, three algorithms

% Work on TA min time
%Optimal-time reachability has been studied extensively for timed automata. 
The earliest work on minimal-time reachability was by Courcoubetis and
Yannakis~\cite{CY92}, who first addressed the problem of computing
lower and upper bounds in timed automata. Several algorithms have been
developed since to improve performance~\cite{NTY00,ZNL16,ZNL16SAC}, by \eg{}
using parallelism.
Related problems have been studied, such as minimal-time reachability for
weighted timed automata~\cite{ATP04}, minimal-cost reachability in priced
timed automata~\cite{BF01}, and job scheduling for timed automata~\cite{AAM06}.

% IM, decidability results
% BLR05

Concerning parametric timed automata, to the best of our knowledge, the
minimal-time reachability problem was not tackled in the past.  The
reachability-emptiness problem (``the emptiness of the parameter valuation set
for which a given set of locations is reachable'') is undecidable~\cite{AHV93},
with various settings considered, notably a single clock compared to
parameters~\cite{Miller00} or a single rational-valued or integer-valued
parameter~\cite{Miller00,BBLS15}
%\ea{here, we can save references and just keep my survey}
%\vb{Not necessary, references don't count for page limit anyway}
(see~\cite{Andre18STTT} for a survey).
Only severely limiting the number of clocks (\eg{}
\cite{AHV93,BO14,BBLS15,AM15}), and often restricting to
integer-valued parameters, can bring some decidability.
Emptiness for the subclass of L/U-PTAs is also decidable~\cite{BLR05}.
Minimizing a parameter can however be considered done in the setting of
upper-bound PTAs (PTAs in which the clocks are only restricted from above): the
exact synthesis of integer valuations for which a location is reachable can be
done~\cite{BlT09}, and therefore the minimum valuation of a parameter can be
obtained.
\todo{say something about rational-value parameters :( (basically, I
cannot do it, although it's very, very closely related to~\cite{SBM14}}
\vb{I have no idea what to do with~\cite{BLR05} and~\cite{SBM14}}\ea{I mention \cite{SBM14} in the conclusion, as it's not strictly a related work. Talking about \cite{BLR05} should be a good idea though.}

% BEGIN LONG VERSION% Overview
\LongVersion{
\paragraph{Overview.}
In \cref{section:preliminaries} we provide preliminaries on TAs and PTAs,
and formalize our problem statements.
\cref{section:theory} addresses the theoretical side of our problems.
\cref{section:minParam} addresses the \emph{parameter minimization synthesis} problem.
% We provide theoretical results and an algorithm that solves the problem.
In \cref{section:minTime} solve the \emph{minimal-time reachability/synthesis}
problems.
We present our experiments in \cref{section:experiments} and conclude in \cref{section:conclusion}.
}
% END LONG VERSION

%%%%%%%%%%%%%%%%%%%%%%%%%%%%%%%%%%%%%%%%%%%%%%%%%%%%%%%%%%%%
%%%%%%%%%%%%%%%%%%%%%%%%%%%%%%%%%%%%%%%%%%%%%%%%%%%%%%%%%%%%
\section{Preliminaries}\label{section:preliminaries}
%%%%%%%%%%%%%%%%%%%%%%%%%%%%%%%%%%%%%%%%%%%%%%%%%%%%%%%%%%%%
%%%%%%%%%%%%%%%%%%%%%%%%%%%%%%%%%%%%%%%%%%%%%%%%%%%%%%%%%%%%

\LongVersion{
%%%%%%%%%%%%%%%%%%%%%%%%%%%%%%%%%%%%%%%%%%%%%%%%%%%%%%%%%%%%
\subsection{Clocks, parameters and guards}
%%%%%%%%%%%%%%%%%%%%%%%%%%%%%%%%%%%%%%%%%%%%%%%%%%%%%%%%%%%%
}

We assume a set~$\Clock = \{ \clock_1, \dots, \clock_\ClockCard \} $ of
\emph{clocks}, \ie{} real-valued variables that evolve at the same rate.  A
clock valuation is\LongVersion{ a function} $\clockval : \Clock \rightarrow
\Rgeqzero$.  We write $\ClocksZero$ for the clock valuation assigning $0$ to
all clocks.  Given $d \in \Rgeqzero$, $\clockval + d$
\ShortVersion{is the valuation}\LongVersion{denotes the valuation} \st{}
$(\clockval + d)(\clock) = \clockval(\clock) + d$, for all $\clock \in \Clock$.
Given $\resets \subseteq \Clock$, we define the \emph{reset} of a
valuation~$\clockval$, denoted by $\reset{\clockval}{\resets}$, as follows:
$\reset{\clockval}{\resets}(\clock) = 0$ if $\clock \in \resets$, and
$\reset{\clockval}{\resets}(\clock)=\clockval(\clock)$ otherwise.

% Throughout this paper,
We assume a set~$\Param = \{ \param_1, \dots, \param_\ParamCard \} $ of
\emph{parameters}\LongVersion{, \ie{} unknown constants}.  A parameter {\em
valuation} $\pval$ is\LongVersion{ a function} $\pval : \Param \rightarrow
\grandqplus$.
%\vb{Why not $\pval : \Param \rightarrow \Rgeqzero$? And parameters cannot be
%0 or negative?}
%\lp{They can be 0 but not negative.}
We denote ${\compOp} \in \{<, \leq, =, \geq, >\}$, ${\compOpLeq}
\in \{<, \leq\}$, and ${\compOpGeq} \in \{>, \geq\}$.  A guard~$\guard$ is a
constraint over $\Clock \cup \Param$ defined by a conjunction of inequalities
of the form $\clock \compOp d$ or $\clock \compOp \param$, with
$\clock\in\Clock$, $d \in \grandn$ and $\param \in \Param$.  Given a guard
$\guard$, we write~$\clockval\models\pval(\guard)$ if the expression obtained
% by replacing each~$\clock\in\Clock\cap\guard$ by~$\clockval(\clock)$ and
% each~$\param\in\Param\cap\guard$ by~$\pval(\param)$ evaluates to true.  
by replacing each clock~$\clock\in\C$ appearing in~$\guard$ by~$\clockval(\clock)$
and each parameter~$\param\in\Param$ appearing in~$\guard$ by~$\pval(\param)$
evaluates to true.

%%%%%%%%%%%%%%%%%%%%%%%%%%%%%%%%%%%%%%%%%%%%%%%%%%%%%%%%%%%%
\subsection{Parametric timed automata}
%%%%%%%%%%%%%%%%%%%%%%%%%%%%%%%%%%%%%%%%%%%%%%%%%%%%%%%%%%%%

\LongVersion{
Parametric timed automata (PTA) extend timed automata with parameters within guards and invariants in place of integer constants~\cite{AHV93}.
}

%----------------------------------------------------------
\begin{definition}[PTA]\label{def:uPTA}
	A PTA $\A$ is a tuple \mbox{$\A = (\Actions, \Loc, \locinit, \Clock, \Param, \invariant, \Edges)$}, where:
    \begin{ienumeration}
		\item $\Actions$ is a finite set of actions,
		\item $\Loc$ is a finite set of locations,
		\item $\locinit \in \Loc$ is the initial location,
% 		\item $\LocsFinal \subseteq \Loc$ is the set of accepting locations,
		\item $\Clock$ is a finite set of clocks,
		\item $\Param$ is a finite set of parameters,
        \item $\invariant$ is the invariant, assigning to every $\loc\in \Loc$ a guard $\invariant(\loc)$,
		\item $\Edges$ is a finite set of edges  $\edge = (\loc,\guard,\action,\resets,\loc')$
		where~$\loc,\loc'\in \Loc$ are the source and target locations, $\action \in \Actions$, $\resets\subseteq \Clock$ is a set of clocks to be reset, and $\guard$ is a guard.
    \end{ienumeration}
\end{definition}
%----------------------------------------------------------

Given a parameter valuation~$\pval$ and PTA~$\A$, we denote by
$\valuate{\A}{\pval}$ the non-parametric structure where all occurrences of a
parameter~$\param\in\Param$ have been replaced by~$\pval(\param)$.
Any structure $\valuate{\A}{\pval}$ is also a \emph{timed automaton}. By
assuming a rescaling of the constants (multiplying all constants in
$\valuate{\A}{\pval}$ by their least common denominator), we obtain an
equivalent (integer-valued) TA\LongVersion{, as defined in \cite{AD94}}.

\LongVersion{
\paragraph{L/U-PTAs}
}

%----------------------------------------------------------
\begin{definition}[L/U-PTA]\label{def:LUPTA} % ~\cite{HRSV02}
	An \emph{L/U-PTA} is a PTA where the set of parameters is partitioned into lower-bound parameters and upper-bound parameters,
    \ie{} parameters that appear in guards and invariants in inequalities of
    the form $\param \compOpLeq \clock$, and of the form $\param \compOpGeq
    \clock$ respectively.
\end{definition}

\LongVersion{
%-%-%-%-%-%-%-%-%-%-%-%-%-%-%-%-%-%-%-%-%-%-%-%-%-%-%-%-%-%-
\paragraph{Concrete semantics of TAs.}
%-%-%-%-%-%-%-%-%-%-%-%-%-%-%-%-%-%-%-%-%-%-%-%-%-%-%-%-%-%-

Let us now recall the concrete semantics of TA.
}

%----------------------------------------------------------
\begin{definition}[Semantics of a TA]
	Given a PTA $\A = (\Actions, \Loc, \locinit, \Clock, \Param, \invariant, \Edges)$,
	and a parameter valuation~\(\pval\),
	the semantics of $\valuate{\A}{\pval}$ is given by the timed transition
    system (TTS) $(\States, \sinit, \flecheRel)$, with:
	\begin{itemize}[topsep=0pt]
		\item $\States = \{ (\loc, \clockval) \in \Loc \times \Rgeqzero^\ClockCard \mid \clockval \models \valuate{\invariant(\loc)}{\pval} \}$,
			% \valuate{\valuate{\invariant(\loc)}{\pval}}{\clockval} \text{ evaluates to true} 
		\LongVersion{
		\item }$\sinit = (\locinit, \ClocksZero) $,
		\item  $\flecheRel$ consists of the discrete and (continuous) delay transition relations:
		\begin{ienumeration}
			\item discrete transitions: $(\loc,\clockval) \longueflecheRel{\edge} (\loc',\clockval')$, %with $\action \in \Sigma$,
				if $(\loc, \clockval) , (\loc',\clockval') \in \States$, and
                there exists $\edge = (\loc,\guard,\action,\resets,\loc') \in
                \Edges$, such that $\clockval'= \reset{\clockval}{\resets}$,
                and $\clockval\models\pval(\guard)$,
				% $\valuate{\valuate{\guard}{\pval}}{\clockval}$ evaluates to true.
			\item delay transitions: $(\loc,\clockval) \longueflecheRel{d} (\loc, \clockval+d)$, with $d \in \Rgeqzero$, if $\forall d' \in [0, d], (\loc, \clockval+d') \in \States$.
		\end{ienumeration}
	\end{itemize}
\end{definition}

Moreover we write $(\loc, \clockval)\longuefleche{(d,\edge)}
(\loc',\clockval')$ for a combination of a delay and discrete transition if
$\exists  \clockval'' :  (\loc,\clockval) \longueflecheRel{d}
(\loc,\clockval'') \longueflecheRel{\edge} (\loc',\clockval')$.
%\lp{It would read better if the arrow was labelled $(d,\edge)$ , \ie{} in the same
%order as the occurrence of delay and discrete transition. (check that it is consistent
%afterwards)}
%\vb{Changed it and made the rest of the paper consistent}

Given a TA~$\valuate{\A}{\pval}$ with concrete semantics $(\States, \sinit,
\flecheRel)$, we refer to the states of~$\States$ as the \emph{concrete states}
of~$\valuate{\A}{\pval}$.  A \emph{run} $\varrun$ of~$\valuate{\A}{\pval}$ is a possibly infinite
%\ea{not necessary: fix it later}\vb{fix what later? We later restrict to finite runs anyway}
alternating sequence of concrete
states of $\valuate{\A}{\pval}$, and pairs of edges and delays, starting from the
initial state $\sinit$ of the form $\state_0, (d_0, \edge_0), \state_1,
\cdots$,
with $i = 0, 1, \dots$, and $d_i \in \Rgeqzero$, $\edge_i \in \Edges$, and
$(\state_i , \edge_i , \state_{i+1}) \in \flecheRel$. The set of all finite
runs over $\valuate{\A}{\pval}$ is denoted by $\Runs(\valuate{\A}{\pval})$.
The \emph{duration} of a finite run $\varrun = \state_0, (d_0, \edge_0), \state_1,
\cdots, \state_i$, is given by $\duration(\varrun) = \sum_{0 \leq j \leq i-1} d_j$.
% Given such a run, the associated \emph{timed word} is $(\action_1, \tau_1), (\action_2, \tau_2), \cdots$, where $\action_i$ is the action of edge~$\edge_{i-1}$, and $\tau_i = \sum_{0 \leq j \leq i-1} d_j$, for $i = 1, 2 \cdots$.\footnote{%
% 	The ``$-1$'' in indices comes from the fact that, following usual conventions in the literature, states are numbered starting from~0 while words are numbered from~1.
% }

Given a state~$\state=(\loc, \clockval)$, we say that $\state$ is reachable
in~$\valuate{\A}{\pval}$ if $\state$ is the last state of a run of
$\valuate{\A}{\pval}$.  By extension, we say that $\loc$ is reachable; and by
extension again, given a set~$\LocsTarget$ of locations, we say that
$\LocsTarget$ is reachable if there exists $\loc \in \LocsTarget$ such that
$\loc$ is reachable in~$\valuate{\A}{\pval}$. The set of all finite runs of
$\valuate{\A}{\pval}$ that reach $\LocsTarget$ is denoted by
$\Reach(\valuate{\A}{\pval},\LocsTarget)$.

\paragraph{Minimal reachability.}

%\todo{say that we describe only the MINIMUM in the entire paper for simplicity
%reason, but that everything can be lifted to maximum (I suppose)}
%For sake of readability, we consider throughout this paper the problem of
%\emph{minimal} (time, parameter valuation) reachability.
%All results and algorithms can be dually converted to the \emph{maximum} problem.
%\ea{OK?}
\vb{I'm not sure if this is correct for maximum-time, I think it's quite a bit
more involved.. at least for the algorithm}
\ea{how come? There is just the management of infinity (but I may be wrong). I leave it up to you to remove, but then the title should be changed, and all occurrences of ``optimum'' (but I don't really see where the problem could be?)}
\vb{Max-time reachability changes the solution from being a shortest path to a
    longest path, \eg{} how to deal with cycles. Of course this is still doable
    with some adaptations, but it's not trivial, so I removed it and we should
just stick to minimal throughout the paper}

As the minimal time may not be an integer, but also the smallest value larger than an integer\footnote{%
	Consider a TA with a transition guarded by $\clock > 1$ from $\loc_0$
    to~$\loc_1$,
	then the minimal duration of runs reaching~$\loc_1$ is not~1 but slightly more.
    %\ea{are you ok with this formulation?}\lp{yes}
},
we define a minimum as either a pair in $\grandqplus \times \{ = , > \}$
or~$\infty$.
The comparison operators function as follows: $(c, =) < \infty$, $(c, >) <
\infty$, and $(c_1, {\succ_1}) < (c_2, \succ_2)$ iff either $c_1 < c_2$ or $c_1
= c_2$, ${\succ_1}$ is ${=}$ and ${\succ_2}$ is ${>}$\footnote{
    When we compute the minimum over a set, we actually calculate its
    infimum and combine the value with either $=$ or $>$ to indicate if the
    value is present in the set.
}.
%\ea{OK?}\lp{yes}

Given a set of locations~$\LocsTarget$, the minimal time reachability
of~$\LocsTarget$ in~$\valuate{\A}{\pval}$, denoted by
$\MinTReach(\valuate{\A}{\pval}, \LocsTarget) = \min\{\duration(\varrun) \mid
\varrun \in \Reach(\valuate{\A}{\pval},\LocsTarget)\}$,
is the minimal duration over all runs of~$\valuate{\A}{\pval}$ reaching
$\LocsTarget$.
%$\MaxTReach$ is defined analogously.
%\ea{quite ad-hoc definition, any significant improvement is welcome}%
%\vb{Formalized it a bit}%
%\ea{Replaced $\pi$ with $\varrun$ to be consistent with the rest}%
\ea{It would be good to say that this problem (for TAs) is decidable, and probably PSPACE-complete but I don't have the reference! Can you help me?}%
\vb{Probably~\cite{CY92} suffices, as I mentioned in the mail conversation}
\todo{Refer to decidability results for TAs}
\vb{We should check this! \cite{NTY00} states that the problem is PSPACE-hard
    (because reachability is PSPACE-complete, \cite{CY92} shows the problem is
    PSPACE-complete for some constraints on the number of clocks.. I don't
    think we can assume that the complete problem is PSPACE-complete from this
paper.. \cite{ATP04} gives some results, though it
mainly focuses on weighted timed automata.. I'm just going to assume that
\cite{CY92} suffices for now}

By extension, given a PTA, we denote by $\PMinTReach(\A, \LocsTarget)$ the
minimal time reachability of~$\LocsTarget$ over all valuations, \ie{}
$\PMinTReach(\A, \LocsTarget) = \min_{\pval} \MinTReach(\valuate{\A}{\pval},
\LocsTarget) $.  As we will be interested in synthesizing the valuations
leading to the minimal time, let us define $\SynthPMinTReach(\A, \LocsTarget) =
\{ \pval \mid \MinTReach(\valuate{\A}{\pval}, \LocsTarget) = \PMinTReach(\A,
\LocsTarget) \}$.
%\ea{OK?}\lp{yes}

%Regarding the \emph{minimal upper-bound time synthesis} problem, we define
%$\PMinUBTReach(\A, \LocsTarget)$ as the minimal \emph{upper-bound} time
%reachability of $\LocsTarget$ over all valuations, \ie{}
%$\PMinUBTReach(\A, \LocsTarget) = \min_{\pval} \MaxTReach(\valuate{\A}{\pval},
%\LocsTarget)$, and the corresponding synthesis is given by
%$\SynthPMinUBTReach(\A, \LocsTarget) = \{ \pval \mid
%\MaxTReach(\valuate{\A}{\pval}, \LocsTarget) = \PMinUBTReach(\A, \LocsTarget)
%\}$.

We will also be interested in minimizing the valuation of a given
parameter~$\param_i$ (without any notion of time) reaching a given location, and
we therefore define $\MinPReach(\A, \param_i, \LocsTarget) = \min_{\pval} \{
\pval(\param_i) \mid \Reach(\valuate{\A}{\pval},\LocsTarget) \neq \emptyset\}$.
%\lp{Added $\neq \emptyset$}.
Similarly, we will be interested in synthesizing \emph{all} valuations leading 
% to the minimal valuation of~$\pval$ reaching~$\LocsTarget$, so let us define
to the minimal valuation of~$\param_i$ reaching~$\LocsTarget$, so let us define
$\SynthMinPReach(\A, \param_i, \LocsTarget) = \{ \pval \mid
\Reach(\valuate{\A}{\pval}, \LocsTarget) \neq \emptyset \land \pval(\param_i) =
\MinPReach(\A, \param_i, \LocsTarget) \}$.
% 	\ea{OK?}\ea{[update] No, it's not. That's actually wrong: just added $\Reach(\valuate{\A}{\pval}, \LocsTarget) \neq \emptyset \land$.}\ea{note: this latter paragraph could also be defined at the beginning of \cref{section:minParam}}
%\todo{Check if it's correct now}
% \lp{I think it is now}

% A finite run is \emph{accepting} if its last state $(\loc, \clockval)$ is such that $\loc \in \LocsFinal$.
% The (timed) \emph{language} $\Lg(\valuate{\A}{\pval})$ is defined to be the set of timed words associated with all accepting runs of~$\valuate{\A}{\pval}$.
% $\{\word \mid \text{there is an accepting run of $\A$ over $\word$}\}$ of timed words.

%%%%%%%%%%%%%%%%%%%%%%%%%%%%%%%%%%%%%%%%%%%%%%%%%%%%%%%%%%%%
\subsection{Computation problems}\label{ss:problems}
%%%%%%%%%%%%%%%%%%%%%%%%%%%%%%%%%%%%%%%%%%%%%%%%%%%%%%%%%%%%

%\vb{Probably remove this subsection in the short version, as it's already
%described in the previous part}%
%\lp{If we have a little room I'd rather leave them here as they stand out.}%
%\ea{agree with Laure}

\defProblem
	{Minimal-time reachability}
	{A PTA~$\A$, a subset $\LocsTarget \subseteq \Loc$ of its locations.}
	{Compute $\PMinTReach(\A, \LocsTarget)$\LongVersion{\ie{} the minimal time
    for which $\LocsTarget$ is reachable for any~$\valuate{\A}{\pval}$}.}

\defProblem
	{Minimal-time reachability synthesis}
	{A PTA~$\A$, a subset $\LocsTarget \subseteq \Loc$ of its locations.}
	{Compute $\SynthPMinTReach(\A, \LocsTarget)$\LongVersion{\ie{} set of all
    parameter valuations~$\pval$ for which $\LocsTarget$ is reachable in
    minimal time in~$\valuate{\A}{\pval}$}.}

%\defProblem
%	{Minimal upper-bound time synthesis}
%	{A PTA~$\A$, a subset $\LocsTarget \subseteq \Loc$ of its locations}
%    {Compute $\SynthPMinUBTReach(\A, \LocsTarget)$}

%\defProblem
%	{Parameter minimization synthesis}
%	{A PTA~$\A$, a subset $\LocsTarget \subseteq \Loc$ of its locations}
%    {Compute $\SynthMinPReach(\A, \LocsTarget)$}

Before addressing the problems defined in \cref{ss:problems}, we will address the slightly different problem of minimal-parameter reachability, \ie{} the minimization of a parameter reaching a given location (independently of time).
We will see in~\cref{lemma:paramtotime} that this problem can also give an answer to the minimal-time reachability (synthesis) problem.

\defProblem
	{Minimal-parameter reachability}
	{A PTA~$\A$, a parameter~$\param$, a subset $\LocsTarget \subseteq \Loc$ of
    the locations of~$\A$.}
	{Compute $\MinPReach(\A, \LocsTarget, \param)$\LongVersion{\ie{} the
    minimal valuation for~$\param$ for which $\LocsTarget$ is reachable for
    any~$\valuate{\A}{\pval}$}.}

\defProblem
	{Minimal-parameter reachability synthesis}
	{A PTA~$\A$, a parameter~$\param$, a subset $\LocsTarget \subseteq \Loc$ of
    the locations of~$\A$.}
	{Synthesize $\SynthMinPReach(\A, \LocsTarget, \param)$\LongVersion{\ie{}
    set of all parameter valuations~$\pval$ for which $\LocsTarget$ is
    reachable for a minimal valuation of~$\param$ in~$\valuate{\A}{\pval}$}.}

% \defProblem{\todo{}}{\todo{}}{\todo{}}

%\todo{only if space: example of both problems on a running example. Ideally the minimal time would be an infinimum, to show the subtlety, and the synthesized associated constraint should be interesting enough.}

%%%%%%%%%%%%%%%%%%%%%%%%%%%%%%%%%%%%%%%%%%%%%%%%%%%%%%%%%%%%
\subsection{Symbolic semantics}
%%%%%%%%%%%%%%%%%%%%%%%%%%%%%%%%%%%%%%%%%%%%%%%%%%%%%%%%%%%%

Let us now recall the symbolic semantics of PTAs (see \eg{} \cite{HRSV02,ACEF09}), that we will use to solve these problems.

\paragraph{Constraints}
We first define operations on constraints.
% In the following, we assume %${\succ} \in \{<, \leq\}$ and
% 	${\compOp} \in \{<, \leq, \geq, >\}$.
A linear term over $\Clock \cup \Param$ is of the form $\sum_{1 \leq i \leq \ClockCard} \alpha_i \clock_i + \sum_{1 \leq j \leq \ParamCard} \beta_j \param_j + d$, with
	$\clock_i \in \Clock$,
	$\param_j \in \Param$,
	and
	$\alpha_i, \beta_j, d \in \grandz$.
% Throughout this paper, $\plterm$ denotes a parametric linear term over
% $\Param$, that is a linear term without clocks (\ie{} $\alpha_i = 0$ for all $1
% \leq i \leq \ClockCard$).
% \vb{Check if $\plterm$ is actually used in the paper.. I guess not}
%
A \emph{constraint}~$\C$ (\ie{} a convex polyhedron) over $\Clock \cup \Param$ is a conjunction of inequalities of the form $\lterm \compOp 0$, where $\lterm$ is a linear term.
%
% A \emph{simple inequality} is of the form $\clock \compOp \param$ or $\clock \compOp d$ (with $d \in \grandqplus$), and a \emph{simple constraint} is a conjunction of simple inequalities.
%
% A \emph{diagonal inequality} is of the form $\clock_i - \clock_j \compOp \plterm$, and a \emph{diagonal constraint} is a conjunction of diagonal inequalities.
$\KFalse$ denotes the false parameter constraint, \ie{} the constraint over~$\Param$ containing no valuation.

Given a parameter valuation~$\pval$, $\valuate{\C}{\pval}$ denotes the constraint over~$\Clock$ obtained by replacing each parameter~$\param$ in~$\C$ with~$\pval(\param)$.
Likewise, given a clock valuation~$\clockval$, $\valuate{\valuate{\C}{\pval}}{\clockval}$ denotes the expression obtained by replacing each clock~$\clock$ in~$\valuate{\C}{\pval}$ with~$\clockval(\clock)$.
% A clock valuation~$\clockval$ \emph{satisfies} constraint~$\valuate{\C}{\pval}$ (denoted by $\clockval \models C[\pval]$) if~$\valuate{\C}{\pval}[w]$ evaluates to true.
We say that %a parameter valuation~
$\pval$ \emph{satisfies}~$\C$,
denoted by $\pval \models \C$,
if the set of clock valuations satisfying~$\valuate{\C}{\pval}$ is non-empty.
Given a parameter valuation $\pval$ and a clock valuation $\clockval$, we denote by $\wv{\clockval}{\pval}$ the valuation over $\Clock\cup\Param$ such that 
for all clocks $\clock$, $\valuate{\clock}{\wv{\clockval}{\pval}}=\valuate{\clock}{\clockval}$
and 
for all parameters $\param$, $\valuate{\param}{\wv{\clockval}{\pval}}=\valuate{\param}{\pval}$.
We use the notation $\wv{\clockval}{\pval} \models \C$ to indicate that $\valuate{\valuate{\C}{\pval}}{\clockval}$ evaluates to true.
We say that $\C$ is \emph{satisfiable} if $\exists \clockval, \pval \text{ s.t.\ } \wv{\clockval}{\pval} \models \C$.
% An \emph{integer point} is $\wv{\clockval}{\pval}$, where $\clockval$ is an integer clock valuation, and $\pval$ is an integer parameter valuation.

We define the \emph{time elapsing} of~$\C$, denoted by $\timelapse{\C}$, as the constraint over $\Clock$ and $\Param$ obtained from~$\C$ by delaying all clocks by an arbitrary amount of time.
That is,
% \[\timelapse{\C} = \{ \wv{\clockval'}{\pval} \mid \clockval \models \valuate{\C}{\pval} \land \forall \clock \in \Clock : \clockval'(\clock) = \clockval(\clock) + d, d \in \grandrplus \}\text{.}\]
\(\wv{\clockval'}{\pval} \models \timelapse{\C} \text{ iff } \exists \clockval : \Clock \to \grandrplus, \exists d \in \grandrplus \text { s.t. } \wv{\clockval'}{\pval} \models \C \land \clockval' = \clockval + d \text{.}\)
% We define the \emph{past} of~$\C$, denoted by $\timepast{\C}$, as the constraint over $\Clock$ and $\Param$ obtained from~$\C$ by letting time pass backward by an arbitrary amount of time (see \eg{} \cite{JLR15}).
% That is,
% \[\wv{\clockval'}{\pval} \models \timepast{\C} \text{ iff } \exists \clockval : \Clock \to \grandrplus, \exists d \in \grandrplus \text { s.t. } \wv{\clockval'}{\pval} \models \C \land \clockval' = \clockval - d \land \forall \clock \in \Clock : \clockval'(\clock) \geq 0 \text{.}\]
Given $\resets \subseteq \Clock$, we define the \emph{reset} of~$\C$, denoted by $\reset{\C}{\resets}$, as the constraint obtained from~$\C$ by resetting the clocks in~$\resets$, and keeping the other clocks unchanged.
% Note that $\reset{\C}{\resets}$ can be written $\C \land \bigwedge_{\clock \in \resets} \clock = 0$. FAUX !!
% BEGIN NORMAL VERSION (not used here)
% We denote by $\projectP{\C}$ the projection of~$\C$ onto~$\Param$, \ie{} obtained by eliminating the clock variables (\eg{} using Fourier-Motzkin~\cite{Schrijver86}).
% END NORMAL VERSION (not used here)
Given a subset~$\Param' \subseteq \Param$ of parameters, we denote by $\project{\C}{\Param'}$ the projection of~$\C$ onto~$\Param'$, \ie{} obtained by eliminating the clock variables and the parameters in $\Param \setminus \Param'$ (\eg{} using Fourier-Motzkin\LongVersion{~\cite{Schrijver86}}).
% BEGIN: NOOOO, because syntactically (on the PDF), that's the same
% We also define $\projectP{\C} = \project{\C}{\Param}$, \ie{} eliminating only the clock variables.
% END: NOOOO, because syntactically (on the PDF), that's the same
Therefore, $\projectP{\C}$ denotes the elimination of the clock variables only, \ie{} the projection onto~$\Param$.
Given~$\param$, we denote by $\GetMin(\C, \param)$ the minimum of~$\param$
in a form $(c, \succ)$.  Technically, $\GetMin$ can be implemented using
polyhedral operations as follows: $\project{\C}{\{\param\}}$ is computed, and
then the infimum is extracted; then the operator in $\{ = , > \}$ is inferred
depending whether $\project{\C}{\{\param\}}$ is bounded from below using a
closed or an open constraint.
%\ea{is that clear? That is almost the way I do in
%\imitator{}, except that I remove upper bounds if $>$, so as to obtain a
%constraint $\param > c$}
%\lp{it is clear to me.}
We extend $\GetMin$ to accommodate clocks, thus $\GetMin(\C,\clock)$ returns
the minimal clock value that $\clock$ can take, while conforming to $\C$.
%\todo{Check definition of $\GetMin$ for clocks}
%We define $\GetMax$ analogously for parameters and clocks.
%\vb{Extended $\GetMin$ for clocks and defined $\GetMax$, both needed for PQ
%algorithm}

\LongVersion{
	\paragraph{Symbolic semantics}
}
% %----------------------------------------------------------
% \begin{definition}[Symbolic state]
    A symbolic state is a pair $(\loc, \C)$ where $\loc \in \Loc$ is a
    % location, and $\C$ its associated parametric zone.
    location, and $\C$ its associated constraint, called \emph{parametric zone}.
    %A symbolic state is a pair $(\loc, \C)$ with location $\loc \in \Loc$ and
    %parametric zone $\C$.
% \end{definition}
% %----------------------------------------------------------
% 
% A symbolic state $\state = (\loc, \C)$ is $\pval$-compatible if $\pval \models \C$.

%----------------------------------------------------------
\begin{definition}[Symbolic semantics]\label{def:PTA:symbolic}
	Given a PTA $\A = (\Actions, \Loc, \locinit, \Clock, \Param, \invariant, \Edges)$,
	the symbolic semantics of~$\A$ is defined by the labelled transition system
    called the \emph{parametric zone graph}
	$ \PZG = ( \Edges, \SymbState, \symbstateinit, \symbtrans )$, with
	\begin{itemize}[topsep=0pt]
		\item $\SymbState = \{ (\loc, \C) \mid \C \subseteq \invariant(\loc) \}$, % \in \Loc \times \grandrplus^\ClockCard 
			\LongVersion{
		\item }$\symbstateinit = \big(\locinit, \timelapse{(\bigwedge_{1 \leq i\leq\ClockCard}\clock_i=0)} \land \invariant(\loc_0) \big)$,
				and
		\item $\big((\loc, \C), \edge, (\loc', \C')\big) \in \symbtrans $ if
            $\edge = (\loc,\guard,\action,\resets,\loc')$ and \\
            \mbox{$\C' = \timelapse{\big(\reset{(\C \land \guard)}{\resets}\land
        \invariant(\loc')\big )} \land \invariant(\loc')$}
			with $\C'$ satisfiable.
	\end{itemize}

\end{definition}
%----------------------------------------------------------

That is, in the parametric zone graph, nodes are symbolic states, and arcs are labeled by \emph{edges} of the original PTA.
\LongVersion{%

}Given $\symbstate = (\loc, \C)$, if $\big((\loc, \C), \edge, (\loc', \C')\big) \in \symbtrans $, we write $\Succ(\symbstate, \edge) = (\loc', \C')$.
By extension, we write $\Succ(\symbstate)$ for $\cup_{\edge \in \Edges} \Succ(\symbstate, \edge)$.
\LongVersion{%
	Given $\big(\symbstate, \edge, \symbstate'\big) \in \symbtrans $, we also write $\symbstate \Fleche{\edge} \symbstate'$.
}%
\ShortVersion{Well-known results (see \cite{HRSV02}) connect the concrete and the symbolic semantics.}

% BEGIN LONG VERSION
\LongVersion{

Given a concrete (respectively symbolic) run $(\locinit, \clockval^0)
\longuefleche{d_0, \edge_0} (\loc_1, \clockval^1) \longuefleche {d_1, \edge_1}
\cdots \longuefleche{d_{m-1}, \edge_{m-1,}} (\loc_m, \clockval^m)$
(respectively  $(\loc_0, \C_0) \Fleche{\edge_0} (\loc_1, \C_1) \Fleche
{\edge_1} \cdots \Fleche{\edge_{m-1}} (\loc_m, \C_m)$),
we define the corresponding \emph{discrete sequence} as
$\loc_0  \Fleche{\edge_0} \loc_1 \Fleche {\edge_1} \cdots \Fleche{\edge_{m-1}} \loc_m $.
Two runs (concrete or symbolic) are said to be \emph{equivalent} if their associated discrete sequences are equal.

The following results (proved in, \eg{} \cite{HRSV02}) connect the concrete and the symbolic semantics.

%------------------------------------------------------------
\begin{lemma}\label{prop:run-equivalence}
	Let $\A$ be a PTA, and let $\varrun$ be a run of~$\A$ reaching $(\loc, \C)$.
	Let $\pval$ be a parameter valuation.
	There exists an equivalent run in~$\valuate{\A}{\pval}$ iff $\pval \models \projectP{\C}$.
\end{lemma}
\begin{proof}
    From \cite[Propositions 3.17 and 3.18]{HRSV02}.%\ea{mouais ?}\dl{c'est pas trop loin de trucs de JLR15 aussi sinon}
\end{proof}
%------------------------------------------------------------

%------------------------------------------------------------
\begin{lemma}\label{prop:run-nonparam-param}
	Let $\A$ be a PTA, let $\pval$ be a parameter valuation. %, and $\clockval$ a clock valuation.
	Let $\varrun$ be a run of~$\valuate{\A}{\pval}$ reaching $(\loc, \clockval)$.
	\LongVersion{
	
	}Then there exists an equivalent symbolic run in~$\A$ reaching $(\loc, \C)$, with $\pval \models \projectP{\C}$.
\end{lemma}
\begin{proof}
	From \cite[Proposition 3.18]{HRSV02}.
\end{proof}
%------------------------------------------------------------

}
% END LONG VERSION

% BEGIN LONG VERSION
\LongVersion{
%%%%%%%%%%%%%%%%%%%%%%%%%%%%%%%%%%%%%%%%%%%%%%%%%%%%%%%%%%%%
\subsection{Reachability synthesis}
%%%%%%%%%%%%%%%%%%%%%%%%%%%%%%%%%%%%%%%%%%%%%%%%%%%%%%%%%%%%

Our upcoming algorithm \MinParamSynth{} shares some similarities with the
reachability-synthesis algorithm called \Synth{}: this procedure takes as input a PTA~$\A$ and a set of target locations~$\LocsTarget$, and attempts to synthesize all parameter valuations~$\pval$ for which~$\LocsTarget$ is reachable in~$\valuate{\A}{\pval}$.
\Synth{} was formalized in \eg{} \cite{JLR15} and is a procedure that may not terminate, but that computes an exact result (sound and complete) if it terminates.
\Synth{} traverses the \emph{parametric zone graph} of~$\A$.
% 	, which is a potentially infinite extension of the well-known zone graph of TAs (see, \eg{} \cite{JLR15}\LongVersion{ for a formal definition}).\todo{cite \cite{ACEF09} in final version}

}
% ENDLONG VERSION

%%%%%%%%%%%%%%%%%%%%%%%%%%%%%%%%%%%%%%%%%%%%%%%%%%%%%%%%%%%%
%%%%%%%%%%%%%%%%%%%%%%%%%%%%%%%%%%%%%%%%%%%%%%%%%%%%%%%%%%%%
\section{Computability and intractability}\label{section:theory}
%%%%%%%%%%%%%%%%%%%%%%%%%%%%%%%%%%%%%%%%%%%%%%%%%%%%%%%%%%%%
%%%%%%%%%%%%%%%%%%%%%%%%%%%%%%%%%%%%%%%%%%%%%%%%%%%%%%%%%%%%

%%%%%%%%%%%%%%%%%%%%%%%%%%%%%%%%%%%%%%%%%%%%%%%%%%%%%%%%%%%%
\subsection{Minimal-time reachability}
%%%%%%%%%%%%%%%%%%%%%%%%%%%%%%%%%%%%%%%%%%%%%%%%%%%%%%%%%%%%

% \todo{for both PTAs (\textbf{not done yet}) and L/U-PTAs~\cite{HRSV02}}

The following result 
% follows from 
is a consequence of
a monotonicity property of L/U-PTAs~\cite{HRSV02}. %, which gives that
We can safely replace parameters with some constants in order to compute the solution
to the minimal-time reachability problem, which reduces to the minimal-time
reachability in a TA, which is PSPACE-complete~\cite{CY92}.
\ShortVersion{All proofs are given in~\cite{ABPP19report}\todo{replace with arXiv version!}.}

\ea{argh, I don't have the reference! Vincent?}
\vb{Added CY92, perhaps change later}

% \jaco{\cref{proposition:intractability-PMinTime-LUPTAs} requires some explanation: How does \cref{proposition:intractability-PMinTime-LUPTAs} compare to \cref{proposition:intractability-MinTime-L/U}?
%   Again it seems to have the same title. Is one implied by the other?\ea{I
%   don't think so. \cref{proposition:intractability-PMinTime-LUPTAs} is
%   minimal-time reachability, while
%   \cref{proposition:intractability-MinTime-L/U} is minimal-parameter reachability.}
% The difference between \cref{proposition:reachability:MinTime-L/U} and \cref{proposition:intractability-PMinTime-LUPTAs} is nice and deserves some
% textual explanation.}\ea{I just added a comment between both propositions; if it is not enough, I can elaborate more. }

%------------------------------------------------------------
\begin{proposition}[minimal-time reachability for L/U-PTAs]\label{proposition:reachability:MinTime-L/U}
	The minimal-time reachability problem for L/U-PTAs is PSPACE-complete.
\end{proposition}
%------------------------------------------------------------
% BEGIN LONG VERSION
\LongVersion{
\begin{proof}
	We show that the problem reduces to the minimal-time reachability problem for TAs.
	
	Let $\A$ be an L/U-PTA.
	Let $\pvalzeroinf$ denote the valuation assigning every lower-bound parameter (resp.\ upper-bound parameter) in the guards of~$\A$ to~0 (resp.~$\infty$).
	Let $\Azeroinf = \valuate{\A}{\pvalzeroinf}$ denote the structure obtained as follows: any occurrence of a lower-bound parameter is replaced with~0, and any occurrence of a conjunct $\clock \compOpLeq \param$ (where $\param$ is necessarily a upper-bound parameter) is deleted, \ie{} replaced with~$\CTrue$.
	($\clock \compOpLeq \infty$ is always satisfiable, therefore equivalent to~$\CTrue$.)
	Let us show that the minimal-time reachability problem for the L/U-PTA~$\A$ is equivalent to the minimal-time reachability problem for the TA~$\Azeroinf$.\ea{was this problem defined for TAs?}
	
% 	Note that $\infty$ is only used for upper-bound parameters, and therefore gives constraints of the form $\clock < \infty$ or $\clock \leq \infty$, for some clock~$\clock$.
% 	In order to obtain a proper TA, one can safely delete these conjuncts from the guards (in the line of~\cite{HRSV02}), as these conjuncts will always be satisfied.
	
	\begin{itemize}[topsep=0pt]
		\item[$\Rightarrow$]
			Let $d$ be the solution of the minimal-time reachability problem for~$\A$, \ie{} $\PMinTReach(\A, \LocsTarget)$.
			Let us show that $\LocsTarget$ is reachable in~$d$ time units in~$\Azeroinf$.
			
			Recall that $\PMinTReach(\A, \LocsTarget) = \min_{\pval} \MinTReach(\valuate{\A}{\pval}, \LocsTarget) $.
			Let $\pval$ be a\footnote{This valuation is not necessarily
            unique.} valuation for which the minimal time is obtained.
			Let $\varrun$ be a run of~$\valuate{\A}{\pval}$ for which this
            minimal time is obtained.
			
			Let us recall the following monotonicity result for L/U-PTAs.
			Basically, any run of a valuation is also a run of a ``larger'' valuation (\ie{} smaller lower-bound parameters and larger upper-bound parameters).
		
			%------------------------------------------------------------
			\begin{lemma}[\cite{HRSV02}]\label{lemma:HRSV02:prop4.2}
				Let~$\A$ be an L/U-PTA and~$\pval$ be a parameter valuation.
				Let $\pval'$ be a valuation such that
				for each upper-bound parameter~$\param^+$, $\pval'(\param^+) \geq \pval(\param^+)$
				and
				for each lower-bound parameter~$\param^-$, $\pval'(\param^-) \leq \pval(\param^-)$.
				Then any run of~$\valuate{\A}{\pval}$ is a run of $\valuate{\A}{\pval'}$. 
			\end{lemma}
			%------------------------------------------------------------

			Therefore, $\varrun$ is a run of~$\Azeroinf$, and therefore $\LocsTarget$ is reachable in~$d$ time units in~$\Azeroinf$.
			
		\item[$\Leftarrow$]
			Let $d$ be the solution of the minimal-time reachability problem for~$\Azeroinf$, \ie{} $\MinTReach(\Azeroinf, \LocsTarget)$,
			and let us show there exists a parameter valuation~$\pval$ such that $\LocsTarget$ is reachable in~$d$ time units in~$\valuate{\A}{\pval}$.
			
			Let~$\varrun$ be a run of $\Azeroinf$ for which $\LocsTarget$ is reachable in~$d$ time units.
			The result could follow immediately from \cref{lemma:HRSV02:prop4.2}---if only assigning $0$ and~$\infty$ to parameters was a proper parameter valuation.
% 	Now, the ``valuation'' we considered assigns $\infty$ to upper-bound parameters, which is not a regular parameter valuation.
% 	Therefore, let us reconstruct a finite valuation for upper-bound parameters (the 0-valuation for lower-bound parameters can be kept as such).
			From~\cite{HRSV02,BlT09}, if a location is reachable in the TA obtained by valuating lower-bound parameters with~0 and upper-bound parameters with~$\infty$, then there exists a sufficiently large constant~$C$ such that this run exists in $\valuate{\A}{\pval}$ such that $\pval$ assigns 0 to lower-bound and~$C$ to upper-bound parameters.
			Here, we can trivially pick~$d$, as any clock constraint $\clock \leq d$ will be satisfied for a run of duration~$d$.
			Let $\pval$ assign 0 to lower-bound and~$d$ to upper-bound parameters.
			Then, $\varrun$ is a run of $\valuate{\A}{\pval}$.
			Therefore, $\LocsTarget$ is reachable in~$d$ time units in~$\valuate{\A}{\pval}$, which concludes the proof.
	\end{itemize}
    The result finally follows from the fact that minimal-time reachability
    problem for TAs is PSPACE-complete~\cite{CY92}.
    \ea{Vincent, can you add the reference discussed?}
\vb{Added CY92, perhaps change later}
	\qed
\end{proof}
%------------------------------------------------------------

}
% END LONG VERSION

Computing the minimal time for which a location is reached (\cref{proposition:reachability:MinTime-L/U})
does not mean that we are able to compute exactly all valuations for which this location is reachable in minimal time.
In fact, we show that it is not possible in a formalism for which the emptiness of the intersection is decidable---which notably rules out its representation as a finite union of polyhedra.
The proof idea is that representing it in such a formalism would contradict the undecidability of the emptiness problem for (normal) PTAs.

%------------------------------------------------------------
\begin{proposition}[intractability of minimal-time reachability synthesis for L/U-PTAs]\label{proposition:intractability-PMinTime-LUPTAs}
	The solution to the minimal-time reachability synthesis problem for L/U-PTAs cannot be represented in a formalism for which the emptiness of the intersection is decidable.
\end{proposition}
%------------------------------------------------------------
% BEGIN LONG VERSION
\LongVersion{
\begin{proof}[by \emph{reductio ad absurdum}]
	We use a reasoning sharing similarities with \cite{BlT09,JLR15} and with \cref{proposition:intractability-MinTime-PTAs,proposition:intractability-MinTime-L/U}.
	Assume the solution to the minimal-time reachability synthesis problem for L/U-PTAs can be represented in a formalism for which the emptiness of the intersection is decidable.
	
	Assume an arbitrary PTA~$\A$ with an initial location~$\locinit$; assume a given target location~$\locfinal$.
% 	Assume that~$\A$ is such that, if~$\locfinal$ is reachable, then it is exactly in 1 time unit.
	
	Add a new clock~$\clock$ not used in~$\A$ (and never reset); add a new upper-bound parameter~$\param^u$.
	Augment~$\A$ as follows: add a new initial location~$\locinit'$, and a transition guarded with~$\clock = 0$ from~$\locinit'$ to~$\locinit$.
	Add a transition guarded by $\clock = 2 \land \clock < \param^u$ from~$\locinit'$ to a new location~$\locfinal'$.
	Add a transition guarded by $\clock \leq 1$ from~$\locfinal$ to~$\locfinal'$.
 	Make~$\locfinal$ urgent\footnote{%
		An urgent location is a location where time cannot elapse (depicted in dotted yellow in our figures, and which can be encoded using an extra clock).
	}.
	The construction is given in \cref{figure:intractability:LUPTAs}.
	Also, turn~$\A$ into an L/U-PTA as in the proof of \cref{proposition:intractability-MinTime-L/U}:
		for any parameter~$\param'$, any guard of the form
			$\clock \compOpLeq \param'$, $\clock \compOpGeq \param'$, $\clock = \param'$
		with
			$\clock \compOpLeq \param'^u$, $\clock \compOpGeq \param'^l$, $\param'^l \leq \clock \leq \param'^u$, respectively.
	The obtained PTA~$\A'$ made of the parameters set $\{ \param'^l, \param'^u \mid \param' \in \Param \} \cup \{ \param^u \}$ is an L/U-PTA.
	
	Clearly, $\locfinal'$ is reachable in~$\A'$ in time~2 by taking the transition from~$\locinit'$ to~$\locfinal'$, for any valuation~$\pval$ such that $\pval(\param^u) > 2$.
	In addition, it is reachable in~$\A'$ in time~$\leq 1$ for all valuations of~$\param^u$ iff there exists a parameter valuation for which $\locfinal$ is reachable in~$\A$ in $\leq 1$ time unit.
	
	Now, assume the solution to the minimal-time reachability synthesis problem for L/U-PTAs can be represented in a formalism for which the emptiness of the intersection is decidable.
	Let $\K$ be this solution in~$\A'$ for $\LocsTarget = \{ \locfinal' \}$.
	Then, there exists a parameter valuation reaching~$\locfinal$ in~$\A$ in time~$\leq 1$ iff the intersection of $\K$ with $\param^u < 2 \land \bigwedge_{i} \param_i^l = \param_i^u $ is non-empty.
	But since reachability emptiness is undecidable for PTAs over bounded time (typically in $\leq 1$ time unit) \cite[Theorem~17]{ALM18}, this leads to a contradiction.
	Therefore, $\K$ cannot be represented in a formalism for which the emptiness of the intersection is decidable.
	\qed
\end{proof}
%------------------------------------------------------------

%----------------------------------------------------------
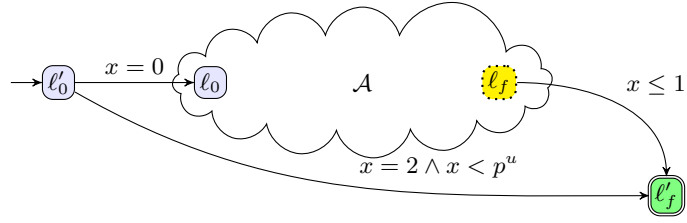
\begin{figure}[tb]
	{\centering
	%------------------------------------------------------------
	\begin{tikzpicture}[->, >=stealth', auto, node distance=2cm, thin]

		% CLOUD AUTOMATON
		\node[location] (l0) at (-2, 0) {$\locinit$};
		\node[location, urgent] (lf) at (+1.8, 0) {$\locfinal$};
		\node[cloud, cloud puffs=15.7, cloud ignores aspect, minimum width=5cm, minimum height=2cm, align=center, draw] (cloud) at (0cm, 0cm) {$\A$};

		\node[location, initial] (l0') at (-4, 0) {$\locinit'$};
		\node[location, final] (lf') at (+4, -1.5) {$\locfinal'$};

		\path
			(l0') edge[] node{$\clock = 0$} (l0)
			(lf) edge[out=0, in=90] node{$\clock \leq 1$} (lf')
			(l0') edge[out=-30, in=180] node{$\clock = 2 \land \clock < \param^u$} (lf')
			;

	\end{tikzpicture}
	%------------------------------------------------------------
	
	}
	\caption{Intractability of minimal-parameter reachability synthesis for
    L/U-PTAs.}
	\label{figure:intractability:LUPTAs}
    \vspace{-1em}
\end{figure}
%----------------------------------------------------------

}
% END LONG VERSION

%%%%%%%%%%%%%%%%%%%%%%%%%%%%%%%%%%%%%%%%%%%%%%%%%%%%%%%%%%%%
\subsection{Minimal-parameter reachability}
%%%%%%%%%%%%%%%%%%%%%%%%%%%%%%%%%%%%%%%%%%%%%%%%%%%%%%%%%%%%

For the full class of PTAs, we will see that these problems are clearly out of reach: if it was possible to compute the solution to the minimal-parameter reachability or minimal-parameter reachability synthesis\LongVersion{ problem}, then it would be possible to answer the reachability emptiness problem---which is undecidable in most settings~\cite{Andre18STTT}.\ea{I'm not entirely sure of this reasoning because, as Vincent noted, we could assume as input that \emph{we know that the target is reachable for some valuations}; then would it really be impossible to compute the minimum…? The issue is that I'm familiar with proving undecidability of emptiness problems, but not really with synthesis (I'm not really sure how to show that a computation is impossible).}

\ea{but what is the decision problem?! there are two computation problems
(``find the minimal parameter'' and ``find the valuations minimizing that parameter'') but I don't see the decision problem. Let's say I'll only address computation problems.}

We first show that an algorithm for the minimal-parameter synthesis problem can
be used to solve the minimal-time synthesis problem,
% 	, then we give an algorithm to directly solve it.
\ie{} the minimal-parameter synthesis problem is harder than the minimal-time synthesis problem.

\ea{slightly modified the proof to talk about problems instead of algorithms (that are not yet defined). Please reread.}

\begin{lemma}[minimal-time from minimal-parameter synthesis]\label{lemma:paramtotime}
    An algorithm that solves the minimal-parameter synthesis problem can be
    used to solve the minimal-time synthesis problem by extending the PTA.
\end{lemma}
\begin{proof}
% 	$\SynthPMinTReach(\A, \LocsTarget)$ min time reach synth
% 	$\SynthMinPReach(\A, \LocsTarget, \param)$ min param reach synth
Assume we are given an arbitrary PTA $\A$, a set of target locations $\LocsTarget$, and
a global clock $\globaltime$ that never resets.
We construct the PTA $\A'$ from $\A$ by adding a new parameter
$\globaltimeparam$, and for every edge $(\loc,\guard,\action,\resets,\loc')$ in
$\A'$ such that $\loc'\in\LocsTarget$, we replace $\guard$ by $\guard \land
\globaltime = \globaltimeparam$.
Note that when a target location from $\LocsTarget$ is reached, we have that
$\globaltime = \globaltimeparam$, hence by minimizing $\globaltimeparam$ we
also minimize $\globaltime$.  Thus, by solving
$\SynthMinPReach(\A',\LocsTarget,\globaltimeparam)$, we effectively solve
$\SynthPMinTReach(\A,\LocsTarget)$.
\end{proof}

The following result states that synthesis of the minimal-value of the parameter is intractable for PTAs.
\begin{proposition}[intractability of minimal-parameter reachability for PTAs]\label{proposition:intractability-MinTime-PTAs}
	The solution to the minimal-parameter reachability for PTAs cannot 
% 		always be represented in a formalism for which the equality is decidable.
be computed in general.
\end{proposition}
%------------------------------------------------------------
% BEGIN LONG VERSION
\LongVersion{
\begin{proof}[by \emph{reductio ad absurdum}]
	Assume the solution to the minimal-parameter reachability for PTAs can be
% 		represented in a formalism for which the equality is decidable.
	computed.
	
	Assume an arbitrary PTA~$\A$ with an initial location~$\locinit$; assume a given target location~$\locfinal$.
	Add a new clock~$\clock$ and a new parameter~$\param$ not used in~$\A$.
	Augment~$\A$ as follows: add a new initial location~$\locinit'$, and a transition guarded with~$\clock = 0$ from~$\locinit'$ to~$\locinit$.
 	Add an unguarded transition from~$\locfinal$ to a new location~$\locfinal''$ resetting~$\clock$, and then a transition guarded by~$\clock = 0 \land \clock = \param$ from~$\locfinal''$ to a new location~$\locfinal'$.
	Add an unguarded transition from~$\locinit'$ to~$\locfinal'$ guarded with $\clock = 1 \land \clock = \param$.
	Let~$\A'$ denote this augmented PTA.
	The construction is given in \cref{figure:intractability:PTAs}.
	
	Clearly, $\locfinal'$ is reachable in~$\A'$ if $\param = 1$.
	In addition, it is reachable in~$\A'$ for $\param = 0$ iff there exists a parameter valuation for which $\locfinal$ is reachable in~$\A$.
	
	Now, assume the solution to the minimal-parameter reachability for~$\A'$ and~$\param$ can be
% 		represented in a formalism for which the equality is decidable.
		computed.
	Let~$\K$ denote this solution (which will typically be $\param = 0$ or $\param = 1$ depending on whether $\locfinal$ is reachable in~$\A$).
	Then, there exists a parameter valuation reaching~$\locfinal$ in~$\A$ iff $\K$ is equal to $\param = 0$.
	But since reachability emptiness is undecidable for PTAs~\cite{AHV93}, this leads to a contradiction.
	Therefore, $\K$ cannot be
% 		represented in a formalism for which equality is decidable.
	computed in general.
		\hfill{}\qed
\end{proof}
%------------------------------------------------------------
}
% END LONG VERSION
% BEGIN SHORT VERSION
\ShortVersion{%
\begin{proof}[sketch]
	By showing that testing equality of ``$\param = 0$'' against the solution
    of the minimal-parameter reachability problem for the PTA in \cref{figure:intractability:PTAs} and $\locfinal'$ is equivalent to solving reachability emptiness of~$\locfinal$ in~$\A$---which is undecidable~\cite{AHV93}.
	Therefore, the solution cannot be computed in general.
% 	\todo{where can we put the long version? I'd prefer \emph{not} ArXiv as it's public. Laure, can you put it on your Web page? If so, please input the URL in the bibtex!}
%     \vb{Otherwise I can put it on GitHub}
%     \ea{Yes, but only if it's on a private URL}
%     \vb{Ok, then I'm just linking to a shared Dropbox file}
\end{proof}
}
% END SHORT VERSION

%----------------------------------------------------------
\begin{figure}[tb]
	{\centering
	%------------------------------------------------------------
	\begin{tikzpicture}[->, >=stealth', auto, node distance=2cm, thin]

		% CLOUD AUTOMATON
		\node[location] (l0) at (-2, 0) {$\locinit$};
		\node[location] (lf) at (+1.8, 0) {$\locfinal$};
		\node[cloud, cloud puffs=15.7, cloud ignores aspect, minimum width=5cm, minimum height=2cm, align=center, draw] (cloud) at (0cm, 0cm) {$\A$};

		\node[location, initial] (l0') at (-4, 0) {$\locinit'$};
		\node[location] (lf'') at (+4, 0) {$\locfinal''$};
		\node[location, final] (lf') at (+4, -1.5) {$\locfinal'$};

		\path
			(l0') edge[] node{$\clock = 0$} (l0)
            (lf) edge node[xshift=5pt]{$\clock := 0$} (lf'')
			(lf'') edge node{$\clock = 0 \land \clock = \param$} (lf')
			(l0') edge[out=-35, in=180] node{$\clock = 1 \land \clock = \param$} (lf')
			;

	\end{tikzpicture}
	%------------------------------------------------------------
	
	}
	\caption{Intractability of minimal-parameter reachability for PTAs}
	\label{figure:intractability:PTAs}
    \vspace{-1em}
\end{figure}
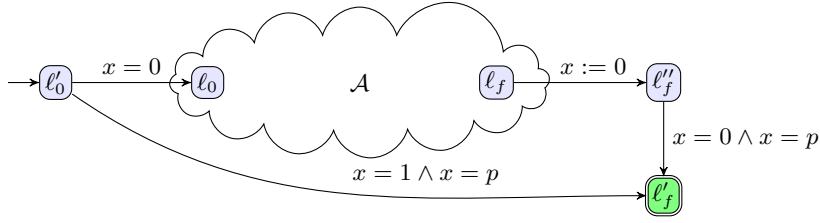
%----------------------------------------------------------

The intractability of minimal-parameter reachability synthesis for PTAs will be implied by the upcoming \cref{proposition:intractability-MinTime-L/U} in a more restricted setting.

% BEGIN LONG VERSION
\LongVersion{
Still, we prove it below with a slightly different condition from \cref{proposition:intractability-MinTime-L/U}.

%------------------------------------------------------------
\begin{proposition}[intractability of minimal-parameter reachability synthesis for PTAs]\label{proposition:intractability-PMinTime-PTAs}
	The solution to the minimal-parameter reachability synthesis for PTAs cannot be represented in a formalism for which the emptiness of the intersection is decidable.
\end{proposition}
%------------------------------------------------------------
\begin{proof}[by \emph{reductio ad absurdum}]
	Assume the solution to the minimal-parameter reachability synthesis for PTAs can be represented in a formalism for which the emptiness of the intersection is decidable.
	
	We use a reasoning similar to that of the proof of~\cref{proposition:intractability-MinTime-PTAs}.
	Assume an arbitrary PTA~$\A$, and augment it into~$\A'$ as in \cref{figure:intractability:PTAs}.
	
	Again, $\locfinal'$ is reachable in~$\A'$ if $\param = 1$.
	In addition, it is reachable in~$\A'$ for $\param = 0$ iff there exists a parameter valuation for which $\locfinal$ is reachable in~$\A$.
	
	Now, assume the solution to the minimal-parameter reachability synthesis for~$\A'$ and~$\param$ can be represented in a formalism for which the emptiness of the intersection is decidable.
	Let~$\K$ denote this solution: note that this solution will either be $\param = 1$ (with all other parameters unconstrained) if $\locfinal$ is unreachable in~$\A$, or a constraint of the form $\param = 0 \land \K'$, for some constraint~$\K'$ over the other parameters.
	Then, there exists a parameter valuation reaching~$\locfinal$ in~$\A$ iff $\K \land \param = 0$ is not empty.
	But since reachability emptiness is undecidable for PTAs~\cite{AHV93}, this leads to a contradiction.
	Therefore, $\K$ cannot be represented in a formalism for which the emptiness of the intersection is decidable.
		\hfill{}\qed
\end{proof}
%------------------------------------------------------------

}
% END LONG VERSION

% BEGIN LONG VERSION
\LongVersion{
Let us now address two subclasses for which the reachability-emptiness problem is decidable:
the class of L/U-PTAs (\cref{ss:minP:LU}), and the class of 1-clock PTAs (\cref{ss:1clock}).
}
% END LONG VERSION

%%%%%%%%%%%%%%%%%%%%%%%%%%%%%%%%%%%%%%%%%%%%%%%%%%%%%%%%%%%%
\paragraph{Intractability of the synthesis for L/U-PTAs.}\label{ss:minP:LU}
%%%%%%%%%%%%%%%%%%%%%%%%%%%%%%%%%%%%%%%%%%%%%%%%%%%%%%%%%%%%

\ea{I feel like it'll be both decidable \emph{and computable} for L/U-PTAs (but I may be wrong!!)}\ea{i'm wrong for full synthesis, see my private document in \texttt{UPTA/}}

\ea{\textbf{it's even not decidable for finding the minimum} :(
Because if it's a L-parameter: that's trivial (0).
But if it's a U-parameter, it amounts at synthesis for U-PTAs with a single parameter---a problem I'm totally unable to solve although I'd say it's very closely related to~\cite{SBM14}
}

% \ea{if I don't study the theoretical aspect of ``Minimal-parameter reachability'', explain why :( Or at least add a remark later}\ea{DONE, in conclusion}

The following result states that synthesis is intractable for L/U-PTAs.
In particular, this rules out the possibility to represent the result using a finite union of polyhedra.

% \jaco{Since L/U-PTA are a special case of PTA, doesn't this result
%   immediately imply proposition 2?\ea{yes, indeed (see my answer above)} Also, without the proof I don't
%   understand the statement. What if the hypothesis of the lemma is
%   false? How does the implication quantify over all L/U PTAs?}\ea{what lemma? what do you mean with ``quantify over all L/U-PTAs''?
%   The issue with these computation problems is that we cannot just say ``we cannot represent the result'', as somehow the L/U-PTA \emph{itself} is a representation of the result :[
%   If this is simpler, we could just write ``we cannot represent it in a finite union of polyhedra''?
%   And in the long version we'd have a more precise technical result.
%   }

%------------------------------------------------------------
\begin{proposition}[intractability of minimal-parameter reachability synthesis for L/U-PTAs]\label{proposition:intractability-MinTime-L/U}
	The solution to the minimal-parameter reachability synthesis for L/U-PTAs cannot always be represented in a formalism for which the emptiness of the intersection is decidable and for which the minimization of a variable is computable.
\end{proposition}
%------------------------------------------------------------
\begin{proof}
	From \cref{lemma:paramtotime,proposition:intractability-PMinTime-LUPTAs}.
		\hfill{}\qed
\end{proof}
%------------------------------------------------------------

% \vb{Force a newpage at this point}
% \clearpage

The minimal-parameter reachability problem remains open for L/U-PTAs (see \cref{section:conclusion}).\LongVersion{

}
Despite these negative results, we will define %in the following
procedures that address not only the class of L/U-PTAs, but in fact the class of full PTAs.
Of course, these procedures are not guaranteed to terminate.\ea{or if we force it, then it is an underapproximation?}

%%%%%%%%%%%%%%%%%%%%%%%%%%%%%%%%%%%%%%%%%%%%%%%%%%%%%%%%%%%%
%%%%%%%%%%%%%%%%%%%%%%%%%%%%%%%%%%%%%%%%%%%%%%%%%%%%%%%%%%%%
\section{Minimal parameter reachability synthesis}\label{section:minParam}
%%%%%%%%%%%%%%%%%%%%%%%%%%%%%%%%%%%%%%%%%%%%%%%%%%%%%%%%%%%%
%%%%%%%%%%%%%%%%%%%%%%%%%%%%%%%%%%%%%%%%%%%%%%%%%%%%%%%%%%%%

%\ea{a.k.a.\ Étienne's algorithm}

\LongVersion{
%%%%%%%%%%%%%%%%%%%%%%%%%%%%%%%%%%%%%%%%%%%%%%%%%%%%%%%%%%%%
\subsection{The algorithm}
%%%%%%%%%%%%%%%%%%%%%%%%%%%%%%%%%%%%%%%%%%%%%%%%%%%%%%%%%%%%

}

We give $\MinParamSynth(\A, \LocsTarget, \param)$ in \cref{algo:MinParamSynth}.
It maintains a set~$\Waiting$ of waiting symbolic states, a set~$\Passed$ of
passed states, a current optimum~$\Kopt$ and the associated optimal
valuations~$\K$.
While $\Waiting$ is not empty, a state is picked in line~\ref{alg:EFS:pick}.
If it is a target state (\ie{} $\loc \in \LocsTarget$) then the projection of
its constraint onto~$\param$ is computed, and the minimum is inferred
(line~\ref{alg:EFS:getmin}).
If that projection improves the known optimum, then the associated parameter
valuations~$\K$ are completely replaced by the one obtained from the current
state (\ie{} the projection of~$\C$ onto~$\Param$).
Otherwise, if $\project{\C}{\{\param\}}$ is equal to the known optimum
(line~\ref{alg:EFS:eq}), then we
add (using disjunction) the associated valuations.
Finally, if the current state is not a target state and has not been visited
before, then we compute its successors and add them to~$\Waiting$ in
lines~\ref{alg:EFS:sucstart}~and~\ref{alg:EFS:sucend}.

Note that if $\Waiting$ is implemented as a FIFO list with ``pick'' the first
element, then this algorithm is a classical BFS procedure.

Also note that if we replace
lines~\ref{alg:EFS:getmin}-\ref{algo:MinParamSynth:update-K} with the statement
$\K \assign \K \lor \projectP{\C}$ (\ie{} adding the parameter valuations to
$\K$ every time the algorithm reaches a target location), we obtain the
standard synthesis algorithm {\Synth} from \eg{}~\cite{JLR15}, that synthesizes all
parameter valuations for which a set of locations is reachable.

%------------------------------------------------------------
\begin{algorithm}[!htb]
	\Input{A PTA $\A$ with symbolic initial state $\symbstateinit = (\locinit,
    \Cinit)$, a set of target locations $\LocsTarget$, a parameter~$\param$.}
	\Output{Constraint $\K$ over the parameters\LongVersion{ solution of
    $\SynthMinPReach(\A, \LocsTarget, \param)$}.}

	\LongVersion{\BlankLine}
	
	\LongVersion{\tcp{Initialization}}
	$\Waiting \assign \{ \symbstateinit \} $ \tcp*{waiting set}
	
	$\Passed \assign \emptyset$ \tcp*{passed set}
	
	$\Kopt \assign \infty$ \tcp*{current optimum}
	
	$\K \assign \KFalse$ \tcp*{current optimum valuations}
	
	\LongVersion{\tcp{Main loop}}
	\While{$\Waiting \neq \emptyset$}{
	
        Pick $\symbstate = (\loc, \C)$ from~$\Waiting$\label{alg:EFS:pick}
		
		$\Waiting \assign \Waiting \setminus \{ \symbstate \} $

		$\Passed \assign \Passed \cup \{ \symbstate \}$
        \nllabel{algo:MinParamSynth:add-to-P}
		
		\uIf(\tcp*[f]{$\symbstate$ is a target state}){$\loc \in \LocsTarget$}{%
			
			$\sopt \assign \GetMin(\C , \param)$
            \tcp*{compute local optimum}\label{alg:EFS:getmin}
			
			\uIf(\tcp*[f]{the optimum is strictly better}){$\sopt < \Kopt$}{
				$\Kopt \assign \sopt$ \tcp*{we found a new best optimum: replace it} 
				
				$\K \assign \projectP{\C}$ \tcp*{completely replace the found
                valuations} \nllabel{algo:MinParamSynth:replace-K}
			}
			
			\ElseIf(\tcp*[f]{the optimum is equal to the one known}){$\sopt =
            \Kopt$}{\label{alg:EFS:eq}
				
				$\K \assign \K \lor \projectP{\C}$ \tcp*{add the found
                valuations} \nllabel{algo:MinParamSynth:update-K}
			}
			
		}
		
		\Else(\tcp*[f]{otherwise explore successors}){%
            \For{\textbf{\emph{each}} $\symbstate' \in
                \Succ(\symbstate)$}{\label{alg:EFS:sucstart}%
% 				\tcp{Compute the successor of $\symbstate$ via $\edge$}
% 				$\symbstate' \assign \Succ(\symbstate,\edge)$
				
				\LongVersion{\tcp{add to waiting list only if not seen before}}
				\lIf{$\symbstate' \notin \Waiting \land \symbstate' \notin
                    \Passed$\nllabel{algo:MinParamSynth:already}}{ % NOTE: removed ``\symbstate'.\C \neq \CFalse \land'' as useless from the very definition of Succ
                    $\Waiting \assign \Waiting \cup \{ \symbstate'
                    \}$\label{alg:EFS:sucend}
				}
			}
		}

	}
	
	\Return $\K$
	\caption{$\MinParamSynth(\A, \LocsTarget, \param)$}
	\label{algo:MinParamSynth}
\end{algorithm}
\begin{figure}[t!]
	{\centering
	%------------------------------------------------------------
	\begin{tikzpicture}[->, >=stealth', auto, xscale=2.5, yscale=1.5]
		% locations
		\node[location, initial] (l1) at (0, 0) {$\loc_1$};
		\node[location] (l2) at (2, 0) {$\loc_2$};
		\node[location, final] (l3) at (2, 1) {$\loc_3$};

        \path
        (l1) edge[bend left] node[xshift=15pt,align=right]{\begin{tabular}{@{} l @{} c} &
        $\clock < \param_1$ \\ $\land$ & $ \clock = 2$\end{tabular}} (l3)
        (l1) edge[] node[xshift=-12pt,align=right]{\begin{tabular}{@{} l @{} c}
        & $\clock < \param_2$ \\ $ \land$ & $ \clock = 1$\end{tabular}}
                node[xshift=-12pt,below]{$\clock := 0$}(l2)
        (l2) edge[bend left] node[align=right]{\begin{tabular}{@{} l @{} c} &
        $\clock = \param_1$ \\ $\land$ & $\clock = 2$\\ $\land$ & $\,\clock >
\param_2$\end{tabular}} (l3)
        (l2) edge[bend right, right] node[align=right]{\begin{tabular}{@{} l
        @{} c} & $\clock = \param_1$ \\ $\land$ & $ \clock = 2$ \\ $\land$ & $
\,\clock = \param_3$\end{tabular}} (l3)
			;
	\end{tikzpicture}
	
	}

	\caption{PTA exemplifying \cref{algo:MinParamSynth}.}
	\label{figure:MinParamSynth:example}
    \vspace{-1em}
\end{figure}
%------------------------------------------------------------

%------------------------------------------------------------
\begin{example}
	Consider the PTA~$\A$ in \cref{figure:MinParamSynth:example}, and
    run $\MinParamSynth(\A, \{ \loc_3 \}, \param_1)$.\ea{the model and results are available in the SVN}
	The initial state is $\symbstate_1 = (\loc_1, \clock \geq 0)$ (we omit the trivial constraints $\param_i \geq 0$).
	Its successors
		$\symbstate_2 = (\loc_3, \clock \geq 2 \land \param_1 > 2)$
	and
		$\symbstate_3 = (\loc_2, \clock \geq 0 \land \param_2 > 1)$
	are added to~$\Waiting$.
	Pick $\symbstate_2$ from~$\Waiting$: it is a target, and therefore $\GetMin(\C_2, \param_1)$ is computed, which gives $(2, >)$.
	Since $(2, >) < \infty$, we found a new minimum, and $\K$ becomes $\projectP{\C_2}$, \ie{} $\param_1 > 2$.
	Pick $\symbstate_3$ from~$\Waiting$: it is not a target, therefore we compute its successors
		$\symbstate_4 = (\loc_3, \clock \geq 2 \land \param_1 = 2 \land 1 < \param_2 < 2)$
	and
		$\symbstate_5 = (\loc_3, \clock \geq 2 \land \param_1 = \param_3 = 2 \land \param_2 > 1)$.
	Pick $\symbstate_4$: it is a target, with $\GetMin(\C_4, \param_1 ) = (2, =)$.
	As $(2, =) < (2, >)$, we found a new minimum, and $\K$ is replaced with $\projectP{\C_4}$, \ie{} $\param_1 = 2 \land 1 < \param_2 < 2$.
	Pick $\symbstate_5$: it is a target, with $\GetMin(\C_4 , \param_1 ) = (2, =)$.
	As $(2, =) = (2, =)$, we found an equally good minimum, and $\K$ is improved with $\projectP{\C_5}$, giving a new $\K$ equal to
	$(\param_1 = 2 \land 1 < \param_2 < 2)
	\lor
	(\param_1 = \param_3 = 2 \land \param_2 > 1)$.
	As $\Waiting = \emptyset$, $\K$ is returned.
	% BEGIN IMITATOR RESULT
% 		p2 > 1
% 	& p3 >= 0
% 	& 2 > p2
% 	& p1 = 2
% 	OR
% 	p2 > 1
% 	& p1 = 2
% 	& p3 = 2
 % END IMITATOR RESULT
\end{example}
\vb{The example seems a bit overly complicated and confusing. (I have some
trouble understanding it)}\ea{Oh. But I wanted to have both a better minimum,
and an identical minimum with $=$ / $>$. And an improved synthesis with the
same minimum.}
\vb{After making the figure a bit more readable it is indeed a good example}
%------------------------------------------------------------

\ea{remark that minimal time reachability can be obtained trivially from minimal-valued reachability by adding a new absolute clock and a new parameter with a guard clock=param on all transitions leading to the target location(s)}

% BEGIN LONG VERSION
\LongVersion{

%%%%%%%%%%%%%%%%%%%%%%%%%%%%%%%%%%%%%%%%%%%%%%%%%%%%%%%%%%%%
\subsection{Correctness}
%%%%%%%%%%%%%%%%%%%%%%%%%%%%%%%%%%%%%%%%%%%%%%%%%%%%%%%%%%%%

\ea{Slightly slightly hand-waving, but please let me know if anything looks unclear/dubious.
In fact, if you are interested, there is actually a (small!) flaw in the algorithm + result + proof, in the case the minimum is of the form $(c,>)$; if you find it, drink is on me (even better, if you have any suggestion for improvement---although I'm afraid it would only make everything more complicated---that's welcome too)
}

%------------------------------------------------------------
\begin{proposition}[soundness]\label{proposition:MinParamSynth:soundness}
	% BEGIN OLD VERSION
% 	Assume $\MinParamSynth(\A, \LocsTarget, \param)$ terminates with result~$\K$.
% 	\begin{enumerate}
% 		\item For any $\pval \models \K$, $\LocsTarget$ is reachable in~$\valuate{\A}{\pval}$.
% 		\item For any $\pval'$, if $\pval'(\param) < \GetMin(\K, \param)$, then $\LocsTarget$ is not reachable in~$\valuate{\A}{\pval'}$.\ea{I think an item is missing (all values must allow the minimum parameter to reach~$\LocsTarget$; perhaps not that important as this is mainly a result holding on polyhedra, coming from $\GetMin$?!}
% 	\end{enumerate}
	% END OLD VERSION
	Assume $\MinParamSynth(\A, \LocsTarget, \param)$ terminates with result~$\K$.
	Let  $\pval \models \K$.
	Then $\pval \models \SynthMinPReach(\A, \LocsTarget, \param)$.
\end{proposition}
\begin{proof}
	Recall that $\SynthMinPReach(\A, \param_i, \LocsTarget) = \{ \pval \mid \Reach(\valuate{\A}{\pval}, \LocsTarget) \neq \emptyset \land \pval(\param_i) = \MinPReach(\A, \param_i, \LocsTarget) \}$.
	Let us first show that $\Reach(\valuate{\A}{\pval}, \LocsTarget) \neq \emptyset$, \ie{} that $\LocsTarget$ is reachable in~$\valuate{\A}{\pval}$.
	From \cref{algo:MinParamSynth}
    (\cref{algo:MinParamSynth:replace-K,algo:MinParamSynth:update-K}), $\K$ is only made of the projection onto~$\Param$ of constraints associated with target symbolic states (\ie{} such that $\loc \in \LocsTarget$).
	Therefore, from \cref{prop:run-equivalence} there exists an equivalent concrete run reaching~$\LocsTarget$ in $\valuate{\A}{\pval}$, which gives that $\Reach(\valuate{\A}{\pval}, \LocsTarget) \neq \emptyset$.
	
	Let us now show that $\pval(\param_i) = \MinPReach(\A, \param_i, \LocsTarget)$.
	First, notice that the entire parametric zone graph of~$\A$ is explored by
    \cref{algo:MinParamSynth}, except when branches are cut (\ie{} successors
    are not explored), \ie{} when a target state is met: in that case, the
    state is added to~$\Passed$ (\cref{algo:MinParamSynth:add-to-P}) but its successors are not computed.
% 		\item when a state is already in the $\Waiting$ (resp.\ $\Passed$) set
    % 		(\cref{algo:MinParamSynth:already}), in which case it will be later (resp.\ has already been) explored.
% 	\end{enumerate}
	Let us show that this result in no loss of information for \MinParamSynth{}
    (in fact, the same holds for \Synth{}, see \eg{} \cite{JLR15}).
	The following result (proved in \eg{} \cite{HRSV02,JLR15}\ea{I think it was in fact never proved, but it's a well-known---and fundamental---result for PTAs (I wrote a formal proof a few years ago in a \LaTeX{} draft that might be somewhere on my computer or on some remote SVN)}) states that the successor of a symbolic state can only restrict the parameter constraint.
	
	%------------------------------------------------------------
	\begin{lemma}\label{lemma:restriction-P-constraint}
		Let $(\loc', \C') \in \Succ((\loc, \C))$.
		Then $\projectP{\C'} \subseteq \projectP{\C}$.
	\end{lemma}
	%------------------------------------------------------------

	From \cref{lemma:restriction-P-constraint}, the unexplored symbolic states do not add any valuation to the known valuation in~$\K$.
	In addition, as \cref{algo:MinParamSynth} iteratively searches for the
    minimal $\Kopt$, then
	\begin{ienumeration}
		\item $\Kopt$ is eventually the minimum of $\param$, and
		\item $\K$ contains all associated parameter valuations associated with~$\param$.
	\end{ienumeration}
	Therefore, $\pval(\param_i) = \MinPReach(\A, \param_i, \LocsTarget)$.
	\qed
\end{proof}
%------------------------------------------------------------

%------------------------------------------------------------
\begin{proposition}[completeness]\label{proposition:MinParamSynth:completeness}
	Assume $\MinParamSynth(\A, \LocsTarget, \param)$ terminates with result~$\K$.
	Let $\pval \models \SynthMinPReach(\A, \LocsTarget, \param)$.
	Then $\pval \models \K$.
\end{proposition}
\begin{proof}
	Recall that $\SynthMinPReach(\A, \param_i, \LocsTarget) = \{ \pval \mid \Reach(\valuate{\A}{\pval}, \LocsTarget) \neq \emptyset \land \pval(\param_i) = \MinPReach(\A, \param_i, \LocsTarget) \}$.
	We use a reasoning dual to \cref{proposition:MinParamSynth:soundness}.
	By definition of $\MinPReach$, $\pval$ is the smallest one for which $\LocsTarget$ is reachable.
	Since $\exists \loc \in \LocsTarget$ reachable in~$\valuate{\A}{\pval}$, from \cref{prop:run-nonparam-param}, there exists an equivalent symbolic run in~$\A$ reaching $(\loc, \C)$, with $\pval \models \projectP{\C}$.
	In addition, from the way the minimum is managed in
    \cref{algo:MinParamSynth} together with the fact that the unexplored states
    do not bring any interesting valuation
    (\cref{lemma:restriction-P-constraint}), then this symbolic state $(\loc,
    \C)$ is kept by \cref{algo:MinParamSynth}, either at
    \cref{algo:MinParamSynth:replace-K} or \cref{algo:MinParamSynth:update-K}, and no further symbolic state will replace it.
	Thus, $\projectP{\C} \subseteq \K$, and therefore $\pval \models \K$.
	\qed
\end{proof}
%------------------------------------------------------------

%}
% END LONG VERSION

%------------------------------------------------------------
\begin{theorem}[correctness]\label{theorem:MinParamSynth:correctness}
	Assume $\MinParamSynth(\A, \LocsTarget, \param)$ terminates with result~$\K$.
	Assume~$\pval$.
	Then $\pval \models \K$ iff $\pval \models \SynthMinPReach(\A, \LocsTarget, \param)$.
\end{theorem}
% BEGIN LONG VERSION
%\LongVersion{
\begin{proof}
	From \cref{proposition:MinParamSynth:soundness,proposition:MinParamSynth:completeness}.
	\qed
\end{proof}
}
% END LONG VERSION
%------------------------------------------------------------

\LongVersion{
%%%%%%%%%%%%%%%%%%%%%%%%%%%%%%%%%%%%%%%%%%%%%%%%%%%%%%%%%%%%
\subsection{A subclass for which the solution can be computed}\label{ss:1clock}
%%%%%%%%%%%%%%%%%%%%%%%%%%%%%%%%%%%%%%%%%%%%%%%%%%%%%%%%%%%%

}
%\jaco{Section 3.4 could be merged with the previous subsection on complexity}

\ShortVersion{
\cref{algo:MinParamSynth} is a semi-algorithm; if it terminates with result
$\K$, then $\K$ is a solution for the $\MinParamSynth$ problem.
Correctness follows from the fact that the algorithm explores the entire
parametric zone graph, except for successors of target states
(from~\cite{HRSV02,JLR15} we have that successors of a symbolic state can only
restrict the parameter constraint, hence we cannot improve).
Furthermore, the minimum is tracked and updated whenever a target state is
reached.
}
% As a decidable subclass, we show that the synthesis can be effectively be achieved for PTAs with a single clock.

We show that synthesis can effectively be achieved for PTAs with a single clock, a decidable
subclass.
%\lp{Rephrased}

%------------------------------------------------------------
\begin{proposition}[synthesis for one-clock PTAs]\label{proposition:decidability:1c}
	The solution to the minimal-parameter reachability synthesis can be computed for 1-clock PTAs using a finite union of polyhedra.
\end{proposition}
%------------------------------------------------------------
% BEGIN LONG VERSION
\LongVersion{
\begin{proof}
	Let us prove termination of \cref{algo:MinParamSynth}.
	In~\cite{AM15}, we showed that the parametric zone graph of a 1-clock PTA is finite.
	By computing successors of symbolic states, \cref{algo:MinParamSynth} clearly explores (a subpart of) the parametric zone graph of~$\A$.
	In addition, no symbolic state is explored twice, thanks to the $\Passed$ set.
	Therefore, \cref{algo:MinParamSynth} terminates for 1-clock PTA and returns a finite union of polyhedra (from the way $\K$ is synthesized).
	The correctness follows from \cref{theorem:MinParamSynth:correctness}.
	\qed
\end{proof}
}
% END LONG VERSION
%------------------------------------------------------------

%%%%%%%%%%%%%%%%%%%%%%%%%%%%%%%%%%%%%%%%%%%%%%%%%%%%%%%%%%%%
%%%%%%%%%%%%%%%%%%%%%%%%%%%%%%%%%%%%%%%%%%%%%%%%%%%%%%%%%%%%
\section{Minimal time reachability synthesis}\label{section:minTime}
%%%%%%%%%%%%%%%%%%%%%%%%%%%%%%%%%%%%%%%%%%%%%%%%%%%%%%%%%%%%
%%%%%%%%%%%%%%%%%%%%%%%%%%%%%%%%%%%%%%%%%%%%%%%%%%%%%%%%%%%%

For minimal-time reachability and synthesis, we assume that the PTA contains a
global clock $\globaltime$ that is never reset.  Otherwise, we 
extend the PTA by simply adding a `dummy' clock $\globaltime$ without any
associated guards or invariants.

\LongVersion{
%%%%%%%%%%%%%%%%%%%%%%%%%%%%%%%%%%%%%%%%%%%%%%%%%%%%%%%%%%%%
\subsection{The algorithm}
%%%%%%%%%%%%%%%%%%%%%%%%%%%%%%%%%%%%%%%%%%%%%%%%%%%%%%%%%%%%
%------------------------------------------------------------
}

\begin{algorithm}[!htb]
	\Input{A PTA $\A$ with symbolic initial state $\symbstateinit = (\locinit,
\Cinit)$, a set of target locations $\LocsTarget$, a global clock that never
resets $\globaltime$.}
\Output{Minimal time $\Topt$ constraint $\K$ over the parameters.}

	\BlankLine
	
	\LongVersion{\tcp{Initialization}}
    $\Queue \assign \{ ( 0, \symbstateinit ) \} $ \tcp*{priority queue ordered by
    time}
	
	$\Passed \assign \emptyset$ \tcp*{passed set}
	
	$\K \assign \KFalse$ \tcp*{current optimum parameter valuations}
	
    $\Topt \assign \infty$ \tcp*{current optimum time}
	
	\LongVersion{\tcp{Main loop}}
	\While{$\Queue \neq \emptyset$}{
	
        $( \Tlow,\symbstate = (\loc, \C)) = \Queue.\Pop()$
        \tcp*{take head of the queue and remove it}\label{alg:PQ:pop}
		
        %	$\Queue \assign \Queue \setminus \{ ( \Tlow,\symbstate ) \} $

		$\Passed \assign \Passed \cup \{ \symbstate \}$ %\tcp*{\vb{consider replacing by visited set}}
		
		\lIf{$\Tlow > \Topt$}{%
            \Break\label{alg:PQ:break}
		}
		\ElseIf(\tcp*[f]{when $\symbstate$ is a target state and $\Tlow\leq\Topt$}){$\loc \in \LocsTarget$}{%
			
			$\K \assign \K \lor \projectP{(\C \land \globaltime =
                \Tlow)}$\tcp*[f]{valuations for which $\Tlow=\Topt$}\label{alg:PQ:addKlow}
%			\uIf{$\minimizeUpperbound$}{%
	
 %               $\Tup \assign \GetMax(\C,\globaltime)$ \tcp*{compute upper time
 %               bound from $\symbstate$}\label{alg:PQ:getmax}

%				$\Tup \assign \sup (\project{\C}{\{\globaltime\}})$ \tcp*{compute upper time bound from $\symbstate$}
			
%				\If(\tcp*[f]{add valuations with $\Tup$ if $\Tup=\Topt$}){$\Tup = \Topt$}{%

%					$\K \assign \K \lor \projectP{(\C \land \globaltime =
 %                   \Tup)}$\label{alg:PQ:addKup}

%				}
%			}
%			\Else(\tcp*[f]{add the valuations with the $\Tlow$ time constraint}) {%
%			}
		% }
		
        }

		\Else(\tcp*[f]{otherwise explore successors}){%
            \For{\textbf{\emph{each}} $\symbstate' \in
                \Succ(\symbstate)$}{\label{alg:PQ:sucstart}%

                \lIf(\tcp*[f]{ignore seen states}){$\symbstate' \in \Queue \lor \symbstate' \in \Passed$}{%
                    \Continue\label{alg:PQ:ignore}
				}

                $\Tlow' \assign \GetMin(\symbstate'.\C,\globaltime)$ \tcp*{get
                minimal time of $\symbstate'.\C$}\label{alg:PQ:lb}
                    
%                $\Tlow' \assign \inf (\project{\symbstate'.\C}{\{\globaltime\}})$ \tcp*{compute lower time bound of $\symbstate'.\C$}

%                \lIf(\tcp*[f]{to-be-optimized time bound}){$\neg\minimizeUpperbound$} {%
%                    $\Tgoal' \assign \Tlow'$\label{alg:PQ:goalmin}
%				}
%
%                \lElse(\tcp*[f]{upper-bound minimization}){%
%                    $\Tgoal' \assign \GetMax(\symbstate'.\C,
%                    \globaltime)$\label{alg:PQ:goalmax}
%					$\Tgoal' \assign \sup (\project{\symbstate'.\C}{\{\globaltime\}})$
%                }

				\If(\tcp*[f]{only add states not exceeding
                    $\Topt$}){$\Tlow' \leq \Topt$}{\label{alg:PQ:leqTopt}

					\If{$\symbstate'.\loc \in \LocsTarget \land \Tlow' <
                        \Topt$}{\label{alg:PQ:bettergoal}

						$\Topt \assign \Tlow'$ \tcp*{new lower time to target}
					}

                    $\Queue.\Push( ( \Tlow', \symbstate' ) )$
                    \tcp*{add to the priority
                    queue\label{alg:PQ:sucend}}\label{alg:PQ:addQ}
				}
			}
		}
	}
	
    \Return $( \Topt, \K )$
	\caption{$\MinTimeSynth(\A, \LocsTarget, \globaltime)$}
	\label{algo:MinTimeSynth}
\end{algorithm}
%------------------------------------------------------------

\lp{Algorithm 2: The else statements at lines 9 and 11 were incorrectly
indented as they seemed to refer to the same if. Check my modification.}
\vb{I changed it back, it's just an `if -- else if -- else' construction, that
should be well known}

\vb{NB: I don't think we discussed $\symbstate'.\C$ earlier}

\ea{we should perhaps say that there is no guarantee that the queue will
enhance the situation! actually i can probably find an easy example that shows
that a queue system should be wrong sometimes (done, that's in the \imitator{}
repository in \texttt{benchmarks/Examples/exPQ-inefficient.imi}. I guess it's
inefficient. Vincent, can you confirm?}

We give $\MinTimeSynth(\A, \LocsTarget, \param)$ in
\cref{algo:MinTimeSynth}.
We maintain a \emph{priority queue} $\Queue$ of waiting symbolic states and order
these by their minimal time (for the initial state this is 0). We further
maintain a set~$\Passed$ of passed states, a current time optimum~$\Topt$
(initially $\infty$), and the associated optimal valuations~$\K$.
We first explain the synthesis algorithm and then the reachability variant.

\paragraph{Minimal-time reachability synthesis.}
While $\Queue$ is not empty, the state with the lowest associated minimal time $\Tlow$
is popped from the head of the queue (line~\ref{alg:PQ:pop}). 
If this time $\Tlow$ is larger than $\Topt$
(line~\ref{alg:PQ:break}), then this also holds for all remaining states in
$\Queue$. Also all successor states from $\symbstate$ (or successors of any
state from $\Queue$) cannot have a better minimal time, thus we can end the
algorithm.

Otherwise, if $\symbstate$ is a target state, we assume that $\Tlow \nless
\Topt$ and thus $\Tlow=\Topt$ (we guarantee this property when pushing states to
the queue). Before adding the parameter valuations to $\K$ in
line~\ref{alg:PQ:addKlow}, we intersect the constraint with $\globaltime =
\Tlow$ in case the clock value depends on parameters, \eg{} if $\C$ is
$\globaltime = \param$.\footnote{
    In case $\Tlow$ is of the form $(c,>)$ with $c\in\grandqplus$, then the intersection of $\C$
    with the linear term $\globaltime = \Tlow$ would result in $\bot$, as the
    exact value $\Tlow$ is not part of the constraint. In
    the implementation, we intersect $\C$ with $\globaltime = \Tlow +
    \varepsilon$, for a small $\varepsilon > 0$.
}

If $\symbstate$ is not a target state, then we consider its successors in
lines~\ref{alg:PQ:sucstart}-\ref{alg:PQ:sucend}. 
We ignore states that have been visited before (line~\ref{alg:PQ:ignore}), and
compute the minimal time of $\symbstate'$ in line~\ref{alg:PQ:lb}.
We compare $\Tlow'$
with $\Topt$ in line~\ref{alg:PQ:leqTopt}. All successor states for
which $\Tlow'$ exceeds $\Topt$ are ignored, as they cannot improve the result.

If $\symbstate'$ is a target state and $\Tlow' < \Topt$, then we update
$\Topt$.
Finally, the successor state is pushed to the priority queue in
line~\ref{alg:PQ:addQ}. Note that we preserve the property that $\Tlow \nless
\Topt$ for the states in $\Queue$.

% UPPER-BOUND
%\paragraph{Minimum upper-bound synthesis.}
%We now assume that $\minimizeUpperbound=\true$. Most of the algorithm remains
%the same so we focus on the differences. Now $\Topt$ describes the best
%upper-bound time. In line~\ref{alg:PQ:break} we have that the algorithm may end
%if $\Tlow$ exceeds $\Topt$, \ie{} when the lower-bound on time for $\symbstate$
%(and every state in $\Queue$) exceeds the best upper-bound, since an
%upper-bound can only be larger.

%We maintain that the priority queue is sorted on the lower-time bound.
%Therefore the best upper-bound may not be at the head of the queue. Whenever we
%find a target location, we compute the upper-bound $\Tup$ at
%line~\ref{alg:PQ:getmax}. We only add the constraint to $\K$ if $\Tup=\Topt$,
%in a similar fashion as described for the lower-bound synthesis.

%When considering successors, we compute at line~\ref{alg:PQ:goalmax} the
%upper-bound time for $\symbstate'$ and set $\Tgoal'$ to this value. Note from
%line~\ref{alg:PQ:leqTopt} that we do not add states with a higher $\Tlow'$ than
%$\Topt$, but it is perfectly fine to have an upper-bound that exceeds $\Topt$
%because an upper-bound may get smaller at successor states. In lines
%\ref{alg:PQ:bettergoal}-\ref{alg:PQ:Kreset} we update $\Topt$ if $\symbstate'$
%is a target state and we found a better upper-bound time. Here, $\K$ may
%be non-empty, thus a reset is required.

%\vb{Note somewhere that parameter minimization cannot be used to solve this
%problem} 

% EARLY TERMINATION / REACHABILITY
\paragraph{Minimal-time reachability.}
When we are interested in just a single parameter valuation, we may end the
algorithm early. The algorithm can be terminated as soon as it reaches
line~\ref{alg:PQ:addKlow}. We can assert at this point that $\Topt$ will not
decrease any further, since all remaining unexplored states have a minimal time
that is larger than or equal to $\Topt$.

%For the minimum upper-bound time reachability, we however cannot terminate
%much earlier as the best upper-bound may continue to decrease after we first
%found the target. After updating $\K$ at line~\ref{alg:PQ:addKup} we can
%end the algorithm earlier if $\Tlow = \Topt$, as no better upper-bound can be
%found at this point.

\LongVersion{
%%%%%%%%%%%%%%%%%%%%%%%%%%%%%%%%%%%%%%%%%%%%%%%%%%%%%%%%%%%%
\subsection{Correctness}
%%%%%%%%%%%%%%%%%%%%%%%%%%%%%%%%%%%%%%%%%%%%%%%%%%%%%%%%%%%%
}

\medskip

\cref{algo:MinTimeSynth} is a semi-algorithm; if it terminates with result
$( \Topt, \K)$, then $\K$ is a solution for the $\MinTimeSynth$ problem.
Correctness follows from the fact that the algorithm explores exactly all
symbolic states in the parametric zone graph that can be reached in at most
$\Topt$ time, except for successors of target states.
Note (again) that successors of a symbolic state can only
restrict the parameter constraint.
Furthermore, $\Topt$ is checked and updated for every encountered successor to
ensure that the first time a target state is popped from the priority
queue $\Queue$, it is reached in $\Topt$ time (after which $\Topt$ never
changes).

%%%%%%%%%%%%%%%%%%%%%%%%%%%%%%%%%%%%%%%%%%%%%%%%%%%%%%%%%%%%
%%%%%%%%%%%%%%%%%%%%%%%%%%%%%%%%%%%%%%%%%%%%%%%%%%%%%%%%%%%%
\section{Experiments}\label{section:experiments}
%%%%%%%%%%%%%%%%%%%%%%%%%%%%%%%%%%%%%%%%%%%%%%%%%%%%%%%%%%%%
%%%%%%%%%%%%%%%%%%%%%%%%%%%%%%%%%%%%%%%%%%%%%%%%%%%%%%%%%%%%

We implemented all our algorithms in the \imitator{} tool~\cite{AFKS12} and
compared their performance with the standard (non-minimization) {\Synth} parameter
synthesis algorithm from~\cite{JLR15}.
%\ea{I don't think this is the correct reference. \cite{ACEF09} is only
%documenting the inverse method, not used here. If you are talking about
%\Synth{}, it is not documented in any paper, but is formalized in \cite{JLR15}.
%As for \MinParamSynth{}, it is supposed to be a contribution of the current
%paper.}
%\vb{Fixed}
For the experiments, we are interested in analysing the performance (in the
form of computation time) of each algorithm, and comparing that with the
performance of standard synthesis.
%\vb{make sure that somewhere there is a line saying ``with standard synthesis
%we mean non-minimization parameter synthesis'' or something in those lines}

\paragraph{Benchmark models.}
We collected PTA models and properties from the \imitator{}
benchmarks library~\cite{Andre18FTSCS}
%(\href{https://www.imitator.fr/library.html}{https://www.imitator.fr/library.html}),
which contains numerous benchmark models from scientific and industrial
domains. We selected all models with reachability properties
%\footnote{
%    We omitted all models from the selection that did not define a simple goal
%    location as the property to be checked, \eg{} instances that check for
%    reachability of the form $(\loc_1 \lor \loc_2) \land (\loc_3 \lor \loc_4)$,
%    or properties that check for the execution order of actions. Note that the
%    time minimization algorithms can handle such models without problems, but
%    it becomes unclear to define parameter constraints such that the
%    time minimization problem can be solved via parameter minimization.\ea{I
%    don't understand this footnote at all: all examples you gave are exactly
%reachability problems, isn't it? Perhaps it's safer to remove it as a whole
%(which also saves space)}
%}
and extended these to include: (1) a new clock variable that represents the
global time $\globaltime$, \ie{} a clock that does not reset, and (2) a new
parameter $\globaltimeparam$ along with the linear term $\globaltime =
\globaltimeparam$ for every transition that targets a goal location, to ensure
that when minimizing $\globaltimeparam$ we effectively minimize $\globaltime$.
In total we have 68 models, and for every experiment we used the extended model
that includes both the global time clock $\globaltime$ and the corresponding
parameter $\globaltimeparam$.

\paragraph{Subsumption.}
For each algorithm that we consider, it is possible to reduce the search space
with the following two reduction techniques:
\begin{itemize}[topsep=0pt]
    \item \emph{State inclusion}~\cite{DT98}: Given two symbolic states
        $\symbstate_1 = (\ell_1,\C_1)$ and $\symbstate_2 = (\ell_2,\C_2)$ with
        $\ell_1=\ell_2$, we say that $\symbstate_1$ is included in
        $\symbstate_2$ if all parameter valuations for
        $\symbstate_1$ are also contained in $\symbstate_2$, \eg{}
        $\C_1$ is $\param>5$ and $C_2$ is $\param > 2$. We may then conclude
        that $\symbstate_1$ is redundant and can be ignored. This check can be
        performed in the successor computation ($\Succ$) to remove included
        states, without altering correctness for minimal-time (or
        parameter) synthesis.
    \item \emph{State merging}~\cite{AFS13atva}: Two states $\symbstate_1 =
        (\ell_1,\C_1)$ and $\symbstate_2 = (\ell_2,\C_2)$  can be merged if
        $\ell_1=\ell_2$ and $\C_1 \cup \C_2$ is a convex polyhedron. The
        resulting state $(\ell_1, \C_1 \cup \C_2)$ replaces $\symbstate_1$ and
        $\symbstate_2$ and is an over-approximation of both states. However,
        reachable locations, minimality, and executable actions are preserved.
\end{itemize}
%\vb{I'm actually starting to wonder if inclusion and merging maybe cause
%erroneous results for minimum-time synthesis..}
State inclusion is a relatively inexpensive computational task and preliminary
results showed that it caused the algorithm to perform equally fast or faster
than without the check. Checking for merging is however a computationally
expensive procedure and thus should not be performed for every newly found
state. For all BFS-based algorithms (standard synthesis and minimal-parameter
synthesis) we merge every BFS layer.
For the minimal-time synthesis algorithm, we empirically studied various merging heuristics
and
found that merging every ten iterations of the algorithm yielded the best
results. We assume that both the inclusion and merging state-space reductions
are used in all experiments (all computation times include the overhead the
reductions), unless otherwise mentioned.
%\vb{check naming of algorithms}

\paragraph{Run configurations.}
For the experiments we used the following configurations:
\begin{itemize}[topsep=0pt]
    \item \ExpMinTimeReach{}: Minimal-time reachability,
    \item \ExpMinTimeSynth{}: Minimal-time synthesis,
%    \item \ExpMinUBTimeSynth{}: Minimal upper-bound time synthesis,
    \item \ExpMinTimeSynthNoIM{}: Minimal time synthesis, without
        reductions,
    \item \ExpMinParamReach{}: Minimal-parameter reachability (of
        $\globaltimeparam$), and
    \item \ExpMinParamSynth{}: Minimal-parameter synthesis (of
        $\globaltimeparam$), and
    \item \ExpSynth{}: Classical reachability synthesis.
%        \ea{this should be
%        ``Classical reachability synthesis''; I let you replace everywhere
%    (note that is called \stylealgo{EFsynth} in like 10 papers in the
%literature, so renaming it here will be very confusing to reviewers.)}
%\vb{Done}
\end{itemize}
%\vb{Perhaps mention algorithm ref or function call}

\paragraph{Experimental setup.}
We performed all our experiments on an Intel$^{\tiny{\text{\textregistered}}}$
Core$^\textsc{tm}$ i7-4710MQ processor with 2.50GHz and 7.4GiB memory, using a
single thread. The six run configurations were executed on each benchmark
model, with a timeout of 3600 seconds.
All our models, results, and information on how to reproduce the results are
available on
\href{https://github.com/utwente-fmt/OptTime-TACAS19}{https://github.com/utwente-fmt/OptTime-TACAS19}.

\subsubsection{Results.}

\begin{figure}[t!]
	\centering
	\begin{subfigure}[b]{0.49\textwidth}
 		\includegraphics[width=\textwidth]{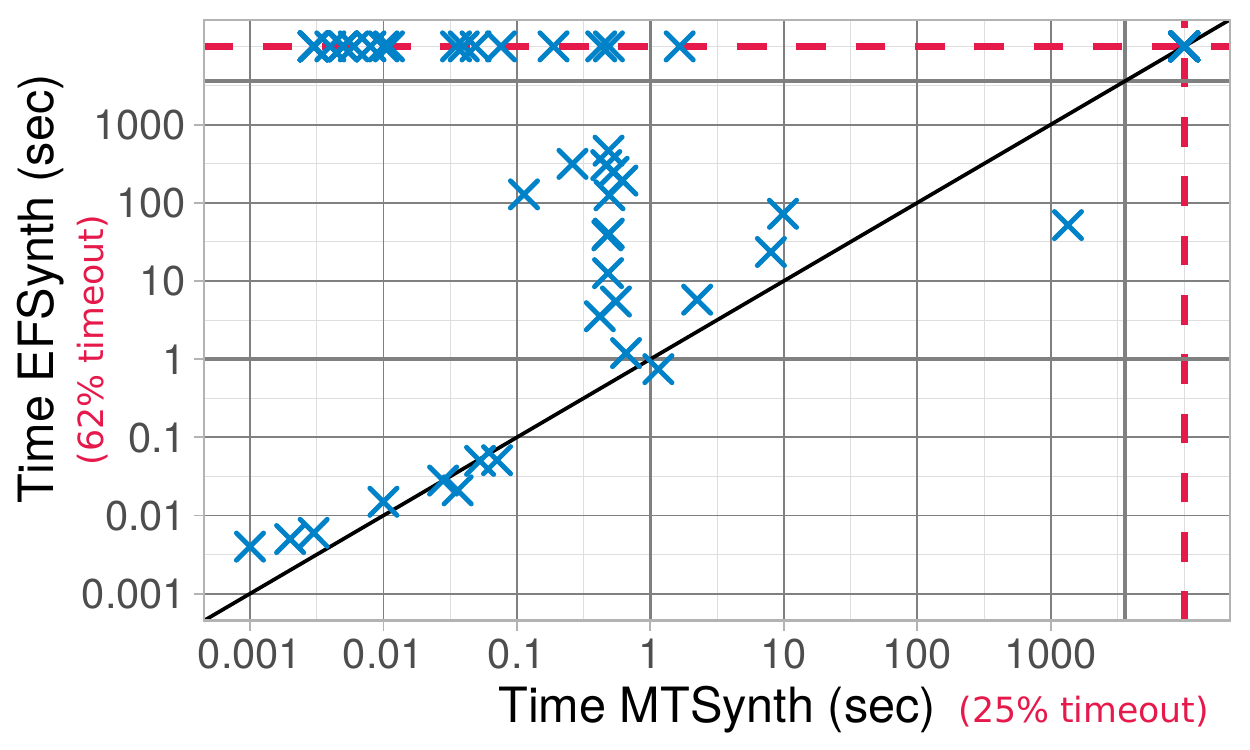}
	\end{subfigure}
	\hfill
	\begin{subfigure}[b]{0.49\textwidth}
 		\includegraphics[width=\textwidth]{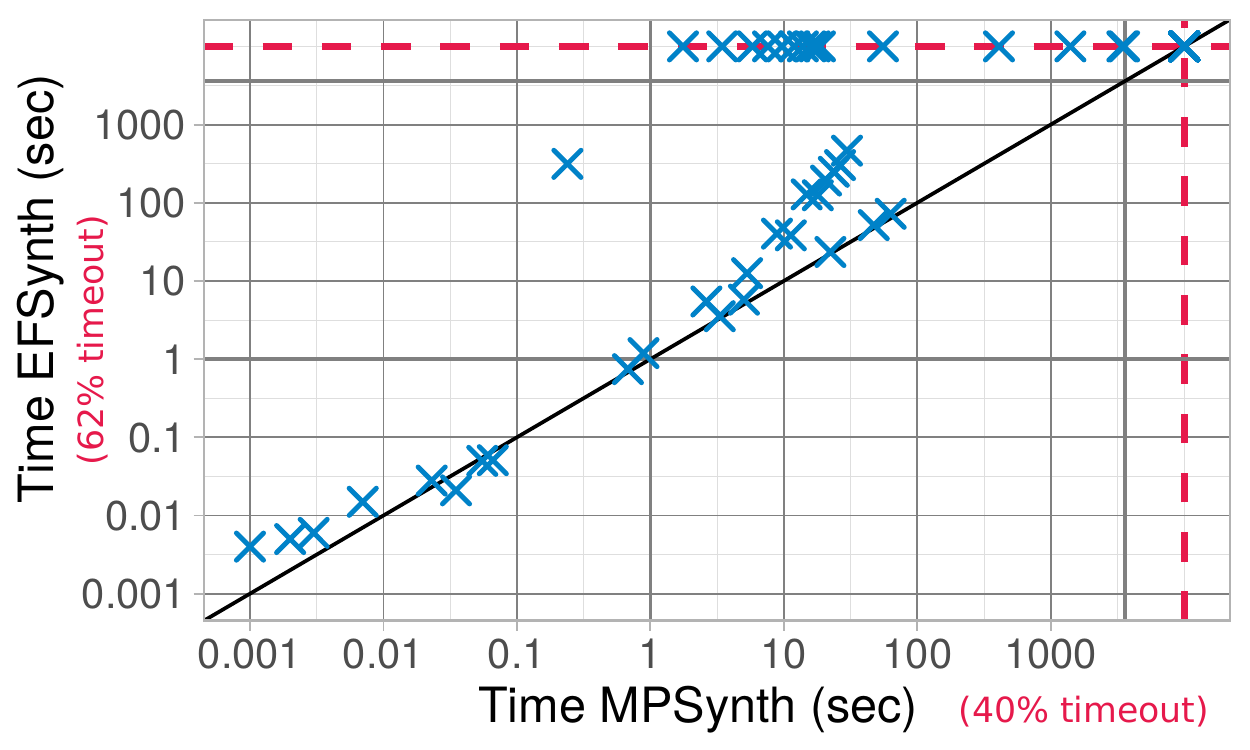}
	\end{subfigure}
    \\
	\begin{subfigure}[b]{0.49\textwidth}
 		\includegraphics[width=\textwidth]{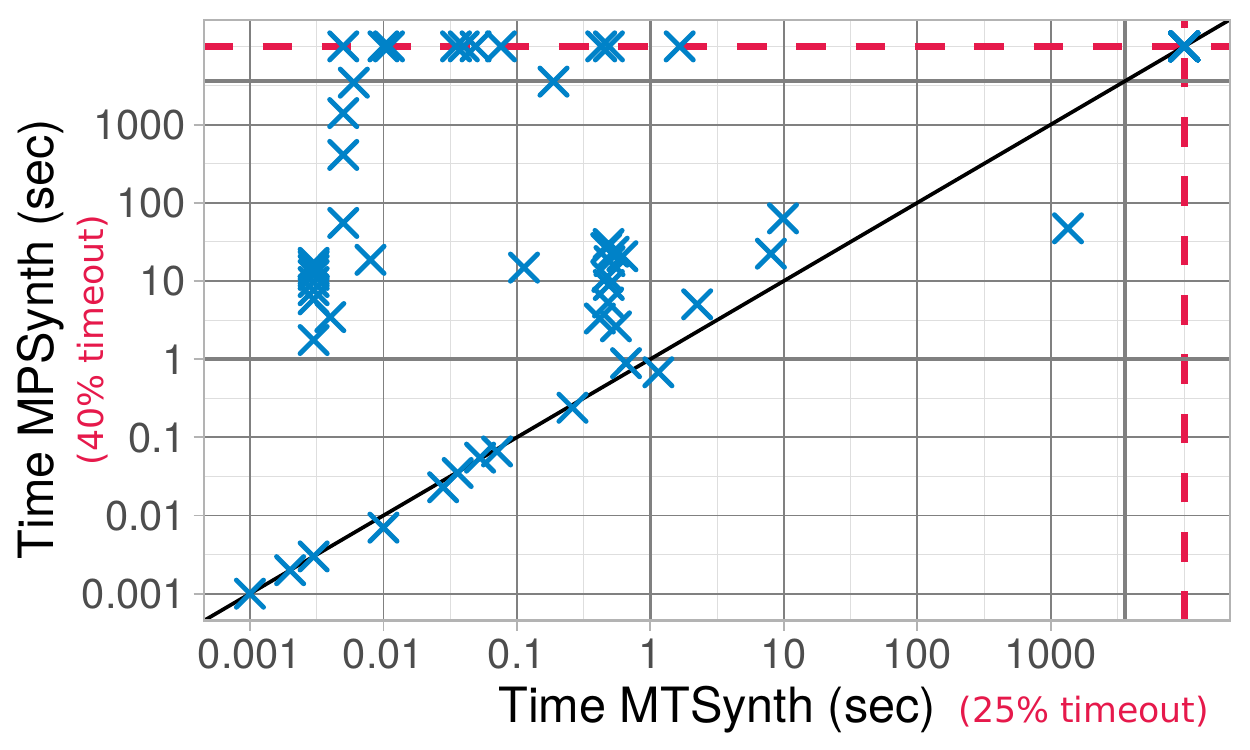}
	\end{subfigure}
	\hfill
	\begin{subfigure}[b]{0.49\textwidth}
 		\includegraphics[width=\textwidth]{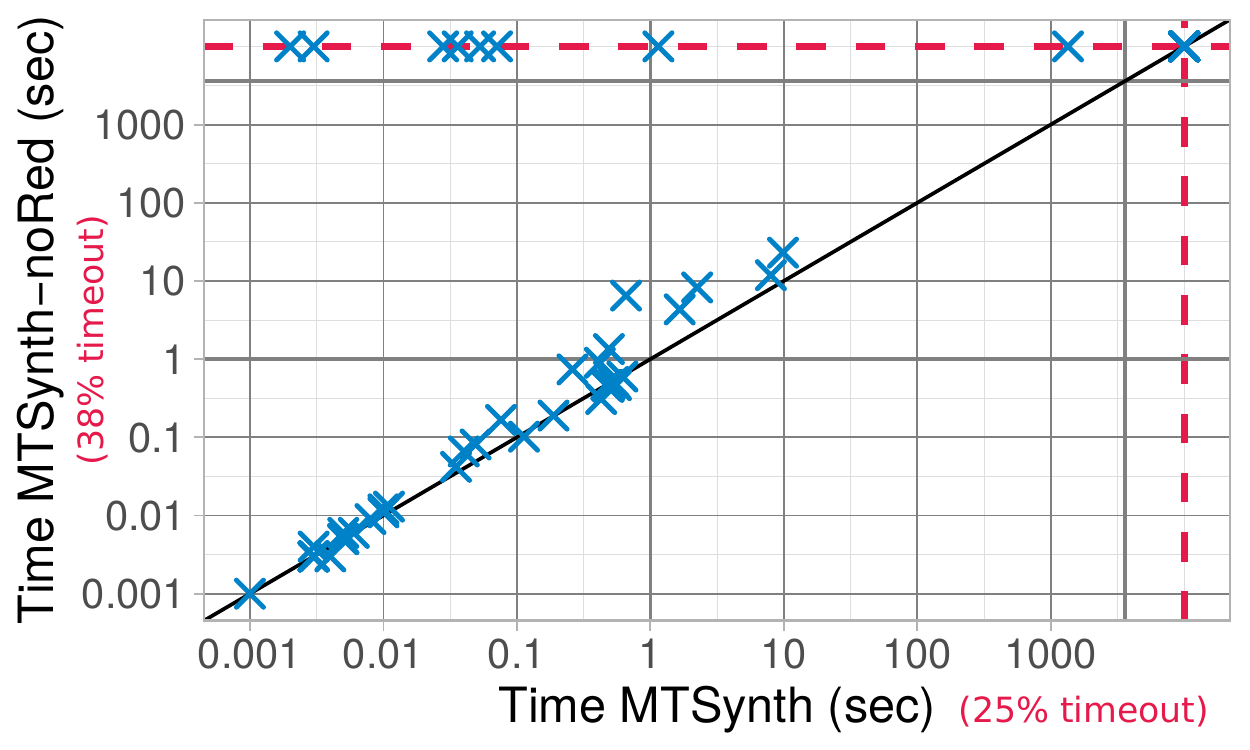}
	\end{subfigure}
	\\
	\begin{subfigure}[b]{0.49\textwidth}
 		\includegraphics[width=\textwidth]{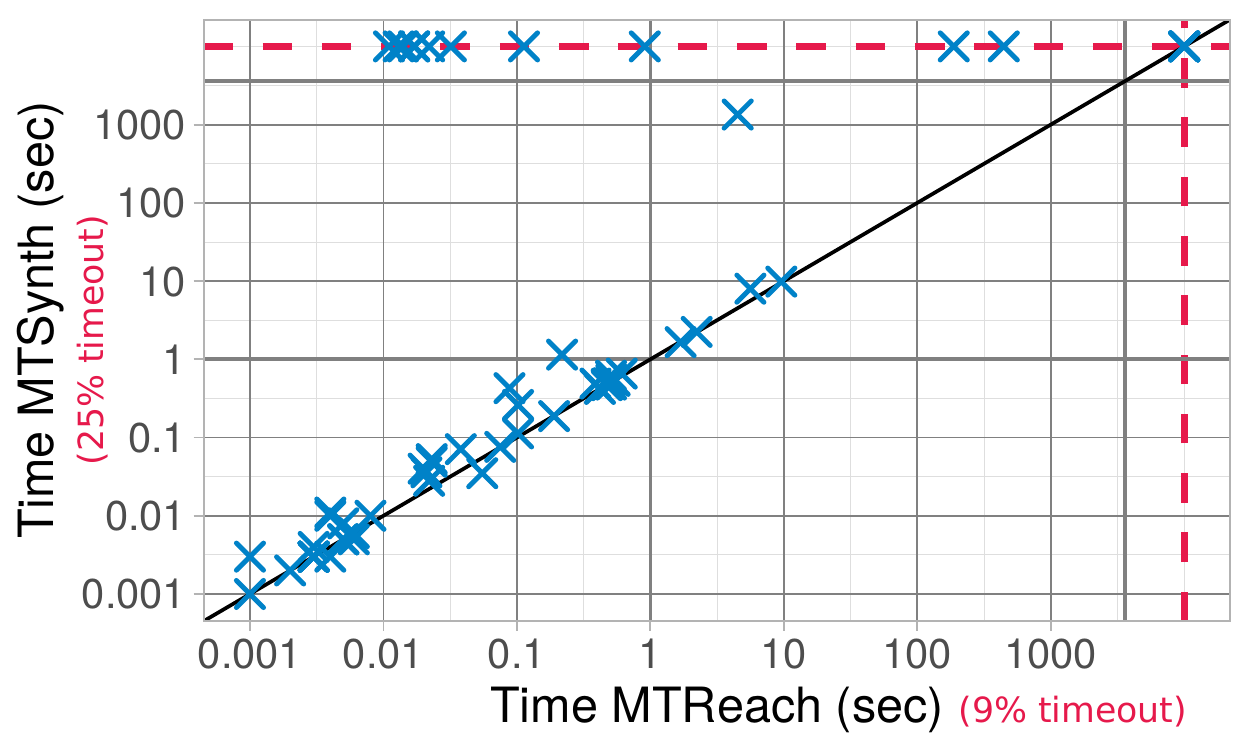}
	\end{subfigure}
	\hfill
	\begin{subfigure}[b]{0.49\textwidth}
 		\includegraphics[width=\textwidth]{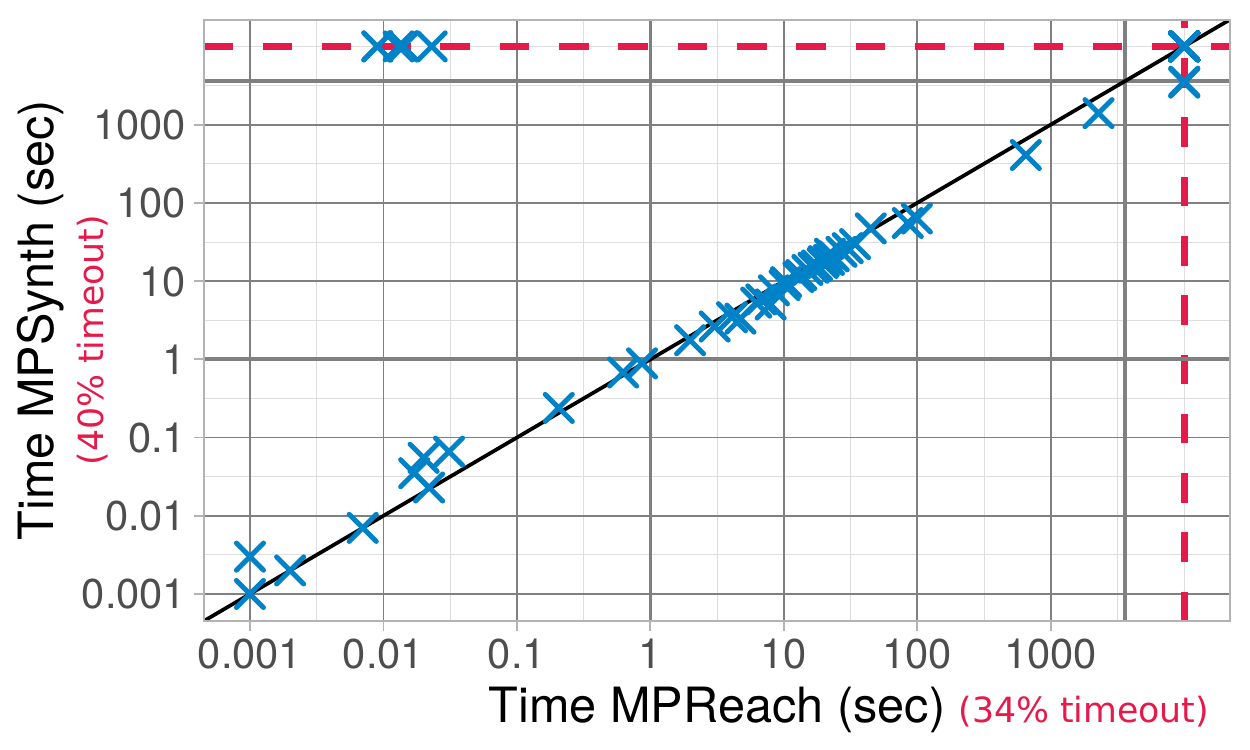}
	\end{subfigure}
    \caption{Scatterplot comparisons of different algorithm configurations. The
        marks on the \textcolor{red}{red} dashed line did not finish computing
    within the allowed time (3600s). }\label{fig:results}
    \vspace{-1em}
\end{figure}

The results of our experiments are displayed in \cref{fig:results}.
%We discuss each graph individually.

{{\ExpMinTimeSynth} vs {\ExpSynth}.}
We observe that for most of the models {\ExpMinTimeSynth} clearly outperforms
{\ExpSynth}. This is to be expected since all states that take more than the
minimal time can be ignored.  Note that the experiments that appear on a
vertical line between $0.1s < x < 1s$ are a scaled-up variant of the same model,
indicating that this scaling does not affect minimal-time synthesis. Finally,
the model plotted at $(1346,52)$ does not heavily modify the
clocks. As a consequence, {\ExpMinTimeSynth} has to explore most of the state
space while continuously having to extract the time constraints, making it
inefficient.

{{\ExpMinParamSynth} vs {\ExpSynth}.}
We can see that {\ExpMinParamSynth} performs more similar to {\ExpSynth} than
{\ExpMinTimeSynth}, which is to be expected as the algorithms differ less.
Still, {\ExpMinParamSynth} significantly outperforms {\ExpSynth}. This is also
because fewer states have to be explored to guarantee optimality (once a
parameter exceeds the minimal value, all its successors can be ignored).

{{\ExpMinTimeSynth} vs {\ExpMinParamSynth}.}
Here, we find that {\ExpMinTimeSynth} outperforms {\ExpMinParamSynth},
similar to the comparison with {\ExpSynth}. The results also show a second
scalable model around $(0.003,10)$ and we see that {\ExpMinParamSynth} is able
to solve the `bad performing model' for {\ExpMinTimeSynth} as quickly as
{\ExpSynth}. Still, we can conclude that the minimal-time
synthesis problem is in general more efficiently solved with the {\ExpMinTimeSynth}
algorithm.

{{\ExpMinTimeSynth} vs {\ExpMinTimeSynthNoIM}.}
Here we can see the advantage of using the inclusion and merging reductions to
reduce the search space. For most models there is a non-existent to slight
improvement, but for others it makes a large difference. While there is some
computational overhead in performing these reductions, this
overhead is not significant enough to outweigh their benefits.

{{\ExpMinTimeReach} vs {\ExpMinTimeSynth}.}
With {\ExpMinTimeReach} we expect faster execution times as the
algorithm terminates once a parameter valuation is found. The experiments
show that this is indeed the case (mostly visible from the timeout line).
However, we also observe that for quite a few models the difference is not as
significant, implying that synthesis results can often be quickly obtained once
a single minimal-time valuation is found.

{{\ExpMinParamReach} vs {\ExpMinParamSynth}.}
Here we also expect {\ExpMinParamReach} to be faster than its synthesis variant.
While it does quickly solve six instances for which {\ExpMinParamSynth} timed
out, other than that there is no real performance gain. We also argue here
that synthesis is obtained quickly when a minimal parameter bound is found.
Of course we are effectively computing a minimal global time, so results may
change when a different parameter is minimized.

\vspace{-.5em}
%%%%%%%%%%%%%%%%%%%%%%%%%%%%%%%%%%%%%%%%%%%%%%%%%%%%%%%%%%%%
%%%%%%%%%%%%%%%%%%%%%%%%%%%%%%%%%%%%%%%%%%%%%%%%%%%%%%%%%%%%
\section{Conclusion}\label{section:conclusion}
%%%%%%%%%%%%%%%%%%%%%%%%%%%%%%%%%%%%%%%%%%%%%%%%%%%%%%%%%%%%
%%%%%%%%%%%%%%%%%%%%%%%%%%%%%%%%%%%%%%%%%%%%%%%%%%%%%%%%%%%%
\vspace{-.5em}

We have designed and implemented several algorithms to solve the minimal-time
parameter synthesis and related problems for PTAs.
%\todo{Mention theoretical results}
From our experiments we observed in general that minimal-time reachability
synthesis is in fact faster to compute compared to standard
synthesis. We further show that synthesis while minimizing a parameter is also
more efficient, and that existing search space reductions apply well to our algorithms.

Aside from the performance improvement, we deem minimal-time reachability synthesis to be
useful in practice. It allows for evaluating which parameter valuations
guarantee that the goal is reached in minimal time.
%We can also synthesise
%parameters to reach the goal in minimum upper-bound time. Both results are 
We consider it particularly valuable when reasoning about real-time systems.

On the theoretical side, we did not address the minimal-parameter reachability problem for L/U-PTAs (we only showed intractability of the synthesis).
While finding the minimal valuation of a given lower-bound parameter is trivial
(the answer is~0 iff the target location is reachable), finding the minimum of
an upper-bound parameter boils down to reachability-synthesis for U-PTAs, a
problem that remains open in general (it is only solvable for integer-valued
parameters~\cite{BlT09}), as well as to shrinking timed automata~\cite{SBM14},
but with 0-coefficients in the shrinking vector---not allowed in~\cite{SBM14}.

A direction for future work is to improve performance by exploiting
parallelism. Parallel random search could % or swarming-based solutions could
significantly speed up the computation process, as demonstrated for timed
automata~\cite{ZNL16SAC,ZNL16}. Another interesting research direction is to
look at maximizing the time to reach the target, or to
minimize the \emph{upper-bound} time to reach the target (\eg{} for minimizing
the worst-case response-time in real-time systems); a preliminary study
suggests that the latter problem is significantly more complex than the
minimal-time synthesis problem.
One may also study other quantitative criteria, \eg
minimizing cost parameters.

%\lp{Future work?}

%\vb{Mention parallel solutions \cite{ZNL16SAC,ZNL16}. And upper-bound
%minimization}

% %%%%%%%%%%%%%%%%%%%%%%%%%%%%%%%%%%%%%%%%%%%%%%%%%%%%%%%%%%%%
% %%%%%%%%%%%%%%%%%%%%%%%%%%%%%%%%%%%%%%%%%%%%%%%%%%%%%%%%%%%%
% \section*{Acknowledgements}
% %%%%%%%%%%%%%%%%%%%%%%%%%%%%%%%%%%%%%%%%%%%%%%%%%%%%%%%%%%%%
% %%%%%%%%%%%%%%%%%%%%%%%%%%%%%%%%%%%%%%%%%%%%%%%%%%%%%%%%%%%%
% XXXXX

%%%%%%%%%%%%%%%%%%%%%%%%%%%%%%%%%%%%%%%%%%%%%%%%%%%%%%%%%%%%%
%%%%%%%%%%%%%%%%%%%%%%%%%%%%%%%%%%%%%%%%%%%%%%%%%%%%%%%%%%%%%
\ifdefined\VersionLong
	\bibliographystyle{alpha} % plain
	\newcommand{\CCIS}{Communications in Computer and Information Science}
	\newcommand{\IJFCS}{International Journal of Foundations of Computer Science}
	\newcommand{\JLAP}{Journal of Logic and Algebraic Programming}
	\newcommand{\LNCS}{Lecture Notes in Computer Science}
	\newcommand{\STTT}{International Journal on Software Tools for Technology Transfer}
	\newcommand{\ToPNoC}{Transactions on Petri Nets and Other Models of Concurrency}
\else
	\bibliographystyle{splncs04} % abbrv
	\newcommand{\CCIS}{CCIS}
	\newcommand{\IJFCS}{IJFCS}
	\newcommand{\JLAP}{JLAP}
	\newcommand{\LNCS}{LNCS}
	\newcommand{\STTT}{STTT}
	\newcommand{\ToPNoC}{ToPNoC}
\fi
\bibliography{OptTime}
%%%%%%%%%%%%%%%%%%%%%%%%%%%%%%%%%%%%%%%%%%%%%%%%%%%%%%%%%%%%%
%%%%%%%%%%%%%%%%%%%%%%%%%%%%%%%%%%%%%%%%%%%%%%%%%%%%%%%%%%%%%

% BEGIN NO APPENDIX
\todo{%

\LongVersion{
\newpage
\appendix

\begin{center}
	\bfseries
	\huge
	Appendix
\end{center}

\ea{ !! Reminder: appendix are NOT allowed at TACAS 2019 !!}

\subsection{New version more practical to reuse/extend}

\ea{below a more convenient version that we may use}

\Synth{} is given in \cref{algo:EF}.

\Synth{} proceeds as a post-order traversal of the symbolic reachability tree, and collects all parametric constraints associated with the target locations~$\LocsTarget$.
Given $\symbstate = (\loc, \C)$, let $\symbstate.\C$ denote~$\C$.
$\Waiting$ denotes the set of states to explore (in a waiting set).
$\Passed$ denotes the set of passed states.
Note that if $\Waiting$ is implemented as a FIFO list with ``pick'' the first element, then this algorithm is a classical BFS procedure.

%------------------------------------------------------------
\begin{algorithm}[htb]
	\Input{A PTA $\A$ with symbolic initial state $\symbstateinit = (\locinit,\Cinit)$, a set of target locations $\LocsTarget$}
	\Output{Constraint $\K$ over the parameters for which $\LocsTarget$ is reachable}

	\BlankLine
	
	\LongVersion{\tcp{Initialization}}
	$\Waiting \assign \{ \symbstateinit \} $ \tcp*{waiting set}
	
	$\Passed \assign \emptyset$ \tcp*{passed set}
	
	$\K \assign \KFalse$ \tcp*{current synthesized constraint}

	\LongVersion{\tcp{Main loop}}
	\While{$\Waiting \neq \emptyset$}{
	
		Pick $\symbstate = (\loc, \C)$ from~$\Waiting$
		
		$\Waiting \assign \Waiting \setminus \{ \symbstate \} $

		$\Passed \assign \Passed \cup \{ \symbstate \}$
		
		\tcp{Target state: store constraint and do not explore successors}
		\lIf{$\loc \in \LocsTarget$}{%
			$\K \assign \K \lor \projectP{\C}$\nllabel{algo:EF:line2}
		}
		\tcp{Otherwise explore successors}
		
		\Else{%
			\For{each $\symbstate' \in \Succ(\symbstate)$}{% % edge
% 			\For{each outgoing edge $\edge$ from $\loc$ in $\A$}{% % edge
% 				\tcp{Compute the successor of $\symbstate$ via $\edge$}
% 				$\symbstate' \assign \Succ(\symbstate,\edge)$
				
				\tcp{Add to waiting list only if not seen before}
				\If{$\symbstate' \notin \Waiting \land \symbstate' \notin \Passed$}{
					$\Waiting \assign \Waiting \cup \{ \symbstate' \}$
				}
			}
		}

	}
	
	\Return $\K$
	\caption{$\Synth(\A, \LocsTarget)$}
	\label{algo:EF}
\end{algorithm}
%------------------------------------------------------------
}

\subsection{Former proof of intractability of minimal-parameter reachability synthesis for L/U-PTAs}

%------------------------------------------------------------
\begin{proposition}[intractability of minimal-parameter reachability synthesis for L/U-PTAs]\label{proposition:intractability-MinTime-L/U:::::old}
	The solution to the minimal-parameter reachability synthesis for L/U-PTAs cannot always be represented in a formalism for which the emptiness of the intersection is decidable and for which the minimization of a variable is computable.\todo{rewrite proof from previous results}
\end{proposition}
%------------------------------------------------------------
% BEGIN LONG VERSION
\LongVersion{
\begin{proof}[by \emph{reductio ad absurdum}]
	Assume a PTA~$\A$ (not L/U), a parameter~$\param$ to minimize and a set~$\LocsTarget$ of locations.
	Transform it into an L/U-PTA as follows (\eg{} using the construction in \cite[Theorem~11]{BlT09}): for any parameter~$\param'$, any guard of the form
		$\clock \compOpLeq \param'$, $\clock \compOpGeq \param'$, $\clock = \param'$
	with
		$\clock \compOpLeq \param'^u$, $\clock \compOpGeq \param'^l$, $\param'^l \leq \clock \leq \param'^u$, respectively.
	The obtained PTA~$\ALU$ made of the parameters set $\{ \param'^l, \param'^u \mid \param' \in \Param \}$ is an L/U-PTA.
	
	Assume it is possible to represent the result of the minimal-parameter reachability synthesis for~$\A'$ in a formalism for which the emptiness of the intersection is decidable and for which the minimization of a variable is computable.
	Compute the minimal valuation of~$\param$ in this result.
	Assume the L/U-PTA~$\ALU'$ where $\param^l$ and $\param^u$ are replaced with this minimal valuation.
	This structure is still an L/U-PTA in $2(|\Param| - 1)$ dimensions.
	Computing $\SynthMinPReach(\ALU, \param, \LocsTarget)$ is equivalent to synthesizing all valuations reaching~$\LocsTarget$ in~$\ALU'$.
	Assume we obtain this result, say~$\K$.
	Now, we will show that if we can test the emptiness of the intersection of the formalism used to represent~$\K$, then we can solve the reachability-emptiness problem for PTAs, which is undecidable.
	
	Consider~$\A'$ which is identical to~$\A$ except that $\param$ is replaced with the minimal valuation computed above.
	This structure is still a general PTA in $|\Param| - 1$ dimensions.
	Let us consider the intersection of~$\K$ with $\param'^l = \param'^u$, for all~$\param'$.
	Then this result is exactly the result of the reachability synthesis for~$\LocsTarget$ in~$\A'$.
	Assume we can test the emptiness of the intersection of~$\K$ with $\param'^l = \param'^u$, for all~$\param'$: we are therefore able to solve the reachability-emptiness problem for PTAs, which is undecidable~\cite{AHV93}---which leads to a contradiction.
		\hfill{}\qed
\end{proof}
%------------------------------------------------------------
}
% END LONG VERSION

}
% END NO APPENDIX

\end{document}